\documentclass[12pt]{article}

\usepackage[margin=1.25in]{geometry}
\usepackage{needspace}

\usepackage{sgamex}
\usepackage{mathptmx}
\usepackage{mathdots}
\usepackage{fancyhdr}
\usepackage{graphicx}
\usepackage{amsmath}
\usepackage{amssymb}
\usepackage{mathrsfs}
\usepackage{amsthm}
\usepackage{bbm}
\usepackage{wrapfig}
\usepackage{multicol}
\usepackage{multirow}
\usepackage{float}
\usepackage{cancel}
 \definecolor{navy}{RGB}{0,0,128}
 \definecolor{mygreen}{RGB}{1, 90, 32}
  
  \definecolor{altblue}{RGB}{0,73,114}

    \definecolor{myred}{rgb}{.75,0,0}
    
     \definecolor{myblue}{rgb}{0.05,0.35,0.9}

\usepackage[colorlinks=true, hyperfootnotes=false, urlcolor=myblue, linkcolor=myblue, citecolor=myblue]{hyperref}

\usepackage[longnamesfirst]{natbib}
\usepackage[utf8]{inputenc}
\usepackage{subfiles}
\usepackage{blindtext}
\usepackage{titlesec}
\usepackage{lipsum}
\usepackage[title,toc,titletoc]{appendix}

\usepackage{sectsty}
\usepackage{dsfont}
\sectionfont{\color{teal}}
\subsectionfont{\color{teal}}
\subsubsectionfont{\color{teal}}
\usepackage[colorinlistoftodos,textsize=tiny]{todonotes}

\usepackage{etoolbox}

\appto\appendix{\addtocontents{toc}{\protect\setcounter{tocdepth}{1}}}

\usepackage{multicol}
\usepackage{enumerate}
\usepackage{setspace} 

\usepackage{caption}
\usepackage{subcaption}
\usepackage{graphicx}

\usepackage{tikz}
\usetikzlibrary{automata} 
\usetikzlibrary{positioning} 
\usetikzlibrary{arrows} 

\usetikzlibrary{patterns}   
\usetikzlibrary{intersections}   
\usetikzlibrary{automata}   
\usetikzlibrary{patterns} 
\usetikzlibrary{arrows,chains,matrix,positioning,scopes}  
\usetikzlibrary{shapes,backgrounds, calc}
\usepackage[all]{xy}
\usepackage[font=footnotesize, textfont=it]{caption}

\newcommand{\RN}[1]{  \textup{\uppercase\expandafter{\romannumeral#1}}}

\newtheoremstyle{break}
  {\topsep}{\topsep}%
  {\itshape}{}%
  {\bfseries}{}%
  {\newline}{}%
\theoremstyle{break}

\usepackage{boiboites}
\usepackage{mdframed}

\newboxedtheorem[boxcolor=orange, background=myblue!5, titlebackground=myblue!20,
   titleboxcolor = black]{theo}{Theorem}{test}

\newtheorem{claim}{Claim}

\newtheorem{definition}{Definition}

\newtheorem*{proposition*}{Proposition}
\newtheorem{theorem}{Theorem}
\newtheorem{observation}{Observation}
\newtheorem*{theorem*}{Theorem}

\newtheorem{lemma}{Lemma}

\newtheorem{assumption}{Assumption}

\theoremstyle{remark}
\theoremstyle{remark}

\newcounter{example}[section]
\newenvironment{example}[1][]{\refstepcounter{example}\par\medskip
   \noindent \textbf{Example~\theexample. #1} \rmfamily}{\medskip $\square$}

\begin{document}
\pagenumbering{gobble}
\title{Robust Performance Evaluation of\\ Independent and Identical Agents\footnote{ An earlier draft appears as the first chapter of my dissertation at the University of Pennsylvania and an abstract appears in \textit{EC '21: Proceedings of the 22nd ACM Conference on Economics and Computation}. I thank my dissertation committee  --- George Mailath, Aislinn Bohren, Steven Matthews, and Juuso Toikka --- for their time, guidance, and encouragement, and Gabriel Carroll for thorough review of an earlier draft. I thank Nageeb Ali, Gorkem Bostanci, Nima Haghpanah, Jan Knoepfle, Rohit Lamba, Natalia Lazzati, Sherwin Lott, Guillermo Ordoñez, Andrew Postlewaite, Doron Ravid, Ilya Segal, Carlos Segura-Rodriguez, Ron Siegel, Ina Taneva, Naomi Utgoff, Rakesh Vohra, Lucy White, Kyle Woodward, and Huseyin Yildirim for helpful comments. Finally, I thank numerous seminar audiences and participants at Seminars in Economic Theory, the 2021 North American Summer Meeting of the Econometric Society, and the 2021 European Economic Association/Econometric Society Meeting.} }

\author{Ashwin Kambhampati\footnote{
	Department of Economics, United States Naval Academy; \url{kambhamp@usna.edu}.}}
\date{\today}
\maketitle
\thispagestyle{empty}

\begin{abstract}

A principal provides nondiscriminatory incentives for independent and identical agents. The principal cannot observe the agents' actions, nor does she know the entire set of actions available to them. It is shown, very generally, that any worst-case optimal contract is nonaffine in performances. In addition, each agent's pay must depend on the performance of another. In the case of two agents and binary output, existence of a worst-case optimal contract is established and it is proven that any such contract exhibits \textit{joint performance evaluation} --- each agent's pay is strictly increasing in the performance of the other. The analysis identifies a fundamentally new channel leading to the optimality of nonlinear team-based incentive pay.   
\end{abstract}

\newpage
\pagenumbering{gobble}
\tableofcontents

\pagenumbering{gobble}

\newpage 
\pagenumbering{arabic}

\begin{quote}

\textit{``The incentive compensation scheme that is ``correct" in one situation will not in general be correct in another. In principle, there could be a different incentive structure for each set of environmental variables. Such a contract would obviously be prohibitively expensive to set up; but more to the point, many of the relevant environmental variables are not costlessly observable to all parties to the contract. Thus, a single incentive structure must do in a variety of circumstances. The lack of flexibility of the piece rate system is widely viewed to be its critical shortcoming: the process of adapting the piece rate is costly and contentious."} \\ --- \cite{NalebuffStiglitz_1983}
\end{quote}

\section{Introduction}
In the canonical moral hazard in teams model, a principal chooses a contract to incentivize a group of agents. Individual actions are unobservable, but stochastically affect observable individual performance. The optimal (Bayesian) contract thus exploits the statistical relationship between actions and performance indicators. 

The canonical model has generated numerous economic insights relevant for policy analysis and management practice (\cite{holmstrom2017pay}). However, it also has some well-known drawbacks. First, in practice, a few common forms of contracts, such as linear contracts and nonlinear bonus contracts, are used in a wide range of scenarios in which there is no compelling reason for there to be a common statistical justification. Second, managers may not possess well-defined prior beliefs over their agents' production environment, limiting the applicability of the model's practical guidance. 

In response to these issues, an emerging literature in contract theory takes a non-Bayesian approach to the moral hazard in teams problem. In this literature, it is assumed that, while the principal may know about some actions her agents can take, she is uncertain about other actions that might be available to them. Hence, she chooses a contract that yields her maximal worst-case expected profits when considering \textit{all} possible unknown actions available to the agents. A general finding is that contracts that are linear in performance indicators provide the best possible profit guarantees (see, for instance, \cite{Carroll_2015}, \cite{DaiToikka_2018}, and \cite{WaltonCaroll_2022}). 

One response to this striking and influential result is that the set of environments the principal considers is too ``large" relative to the Bayesian literature. In practice, a manager may desire her contract to be robust, but also rule out certain production technologies as implausible. For instance, she may \textit{know} her agents are independent and identical, but still want her contract to be robust to all possible technologies the agents may exploit within this class. The independent and identical agents setting is of particular interest not only because it is a benchmark setting in the Bayesian literature, but because existing arguments for linear contracts in the robust contracting literature rely on the ability of an agent to directly influence the performance of others.

Are robust contracts linear if a principal knows her agents are independent and identical? Do they link one agent's pay to the performance of another? This paper shows, very generally, that any nondiscriminatory, worst-case optimal contract is \textit{nonaffine} (and hence, nonlinear) and each agent's pay depends on the performance of another. In the case of two agents and binary output, existence of a worst-case optimal contract  is established and it is proven that any such contract exhibits \textit{joint performance evaluation}, i.e., one agent's pay increases in the performance of the other. This result provides novel foundations for nonlinear team-based incentive pay in the context of numerous existing results in the literature and identifies a channel leading to joint performance evaluation of potential relevance in practice. 

A simple example illustrates the framework and key economic intuition.

\begin{example}\label{simple_ex}
There is a risk-neutral residual claimant (manager) and two identical, risk-neutral agents that perform independent tasks. That is, it is common knowledge that their successes or failures are statistically independent, conditional on the actions they take, and that they cannot influence each other's productivity. Successful completion of a task yields the manager a profit of one and failure yields her a profit of zero. 

The manager knows that each agent can take one of two actions, ``work" or ``shirk". She knows that ``work" results in successful task completion with probability $p_0>0$ at effort cost $c_0 \in (0,p_0)$.  On the other hand, she is uncertain about the effort cost of shirking, $c^* \in \mathbb{R}_+$, and the productivity of shirking, i.e., the probability $p^*<p_0$ with which it results in successful task completion.\begin{footnote}{To be clear, in this example, the principal ``knows" that there is precisely one unknown action. In the baseline model, this hypothesis will be relaxed. In addition, it will no longer be assumed that unknown actions are less productive than known actions.}\end{footnote}

The manager contemplates using one of two contracts, each of which is nondiscriminatory and respects agent limited liability:

\begin{enumerate}

\item Independent Performance Evaluation (IPE): \\Pay each agent $w \in (c_0,1)$ for individual success and $0$ for failure.

\item Nonaffine Joint Performance Evaluation (JPE):\\ Pay each agent a wage $w_0 \in [0,w)$ for individual success and a team bonus \[b = \frac{w-w_{0}}{p_0}\] for joint success. Pay each agent $0$ for failure. Any such contract is calibrated to the contract-action pair $(w, \text{work})$ in the following sense: If an agent succeeds at her task, then her expected wage payment remains $w$ \textit{conditional on the other agent working}. That is,
\[ w_0+ b p_0 =w. \]
\end{enumerate} 

The manager evaluates any contract according to the same criterion. First, for each value of $p^*$, she computes her expected payoff in her preferred Nash equilibrium in the game induced by the contract she offers. Second, she computes the infimum value of her expected payoff over all values of $c^*$ and $p^*$. The resulting payoff is called her \textit{worst-case payoff}.

\begin{figure}
\begin{center}
\renewcommand{\gamestretch}{1.8}
\begin{game}{2}{2}
      & work     & shirk\\
work  & $p_0 w-c_0$ , $p_0 w-c_0$   & $p_0 w- c_0$ , $p^* w-c^*$  \\
shirk  & $p^* w-c^*$ , $p_0 w-c_0$    & $p^* w-c^*$ , $p^*w-c^*$
\end{game}
\end{center}
\caption{Game induced by IPE $w$ given $p^*$.} \label{simple_IPE}
\end{figure}

 Can JPE yield the manager a higher worst-case payoff than IPE? The IPE contract $w$, together with an actual value of $p^*$, induces the game between the agents depicted in Figure~\ref{simple_IPE}. A naïve intuition is that the worst-case scenario for the principal occurs when $p^*=0$; if agents take a shirking action with this success probability, then the principal obtains an expected payoff of zero. But, this logic ignores incentives, as pointed out by \cite{Carroll_2015}. In particular, each agent has a strict incentive to shirk if and only if she obtains a higher expected utility from doing so. Hence, $(\text{work},\text{work})$ is a Nash equilibrium whenever
\[ p^*w-c^* \leq p_0 w-c_0 \iff p^* \leq p_0 -\frac{(c_0-c^*)}{w}, \]  yielding the principal a payoff per agent of \[p_0(1-w).\] The principal's worst-case payoff is instead obtained when $c^*=0$ and as $p^*$ approaches $p_0-\frac{c_0}{w}$ from above. Along this sequence, $(\text{shirk}, \text{shirk})$ is the unique Nash equilibrium and the principal's payoff per agent becomes arbitrarily close to
\[ V_{IPE}(w)= (p_0 -\frac{c_0}{w}) (1-w). \]

\def\stackedpayoffs#1#2{%
\begin{array}{rl} & #1,\\& \hspace{12mm} #2\end{array}
}
\begin{figure}
\begin{center}
\renewcommand{\gamestretch}{1.8}
\begin{game}{2}{2}
      & work     & shirk\\
work  &  $p_0 w-c_0$ , $p_0 w-c_0$   & $p_0 ( w_0+ b p^* )-c_0$, $p^* w-c^*$  \\
shirk  & $p^* w-c^*$, $p_0 ( w_0+ b p^* )-c_0$   & $p^* ( w_0+b p^* )-c^*$, $p^* ( w_0+b p^* )-c^*$
\end{game}
\end{center}
\caption{Game induced by JPE $(w_0, b)$ given $p^*$.} \label{simple_JPE}
\end{figure}

Now, consider the calibrated JPE contract $(w_0, b)$. The game between the agents for a given value of $p^*$ is depicted in Figure \ref{simple_JPE}. Observe that, as under the IPE contract $w$, $(\text{work},\text{work})$ is a Nash equilibrium whenever
\[ p^* \leq p_0 -\frac{c_0}{w}.\]
And, again, the principal's worst-case payoff is obtained when $c^*=0$ and as $p^*$ approaches $p_0-\frac{c_0}{w}$ from above. (Along this sequence, $(\text{shirk}, \text{shirk})$ is the unique Nash equilibrium.) However, a simple calculation shows that the principal obtains a strictly higher worst-case payoff under the calibrated JPE:
\[(p_0 -\frac{c_0}{w}) (1-(w_0+ b p^*))>(p_0 -\frac{c_0}{w}) (1-w)=V_{IPE}(w), \] 
  where the inequality follows from $w_0+ b p^* <w_0+ b p_0 =w$. The intuition is simple. Calibration ensures that worst-case productivity is no lower under the JPE contract than under the IPE contract. But, under the JPE contract, the principal pays agents less in expectation. Each is punished for the shirking of the other. See Figure \ref{JPE_Picture} for an illustration. \begin{figure}
\centering
\includegraphics[scale=0.45]{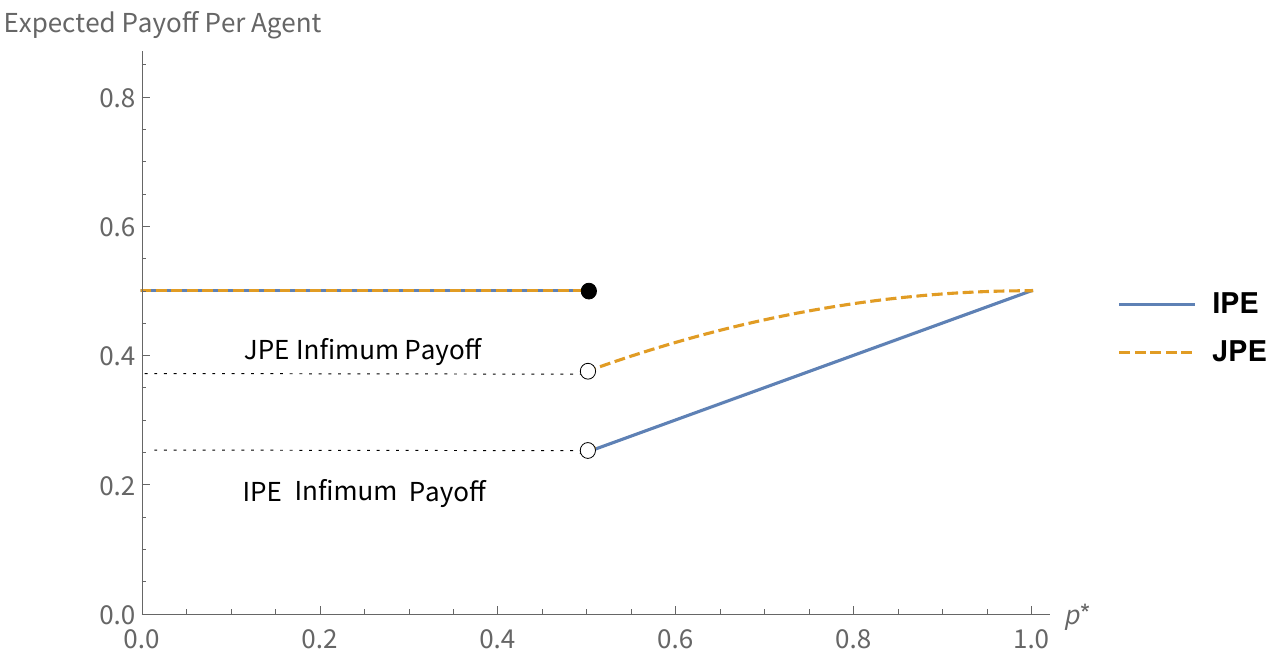}
\caption{Manager's expected payoff per agent as a function of $p^*$ when $c^*=0$. Parameters: $p_0=1$, $c_0=1/4$, $w=1/2$, $w_0=0$, and $b=1/2$.}\label{JPE_Picture}
\end{figure} \end{example}

  Example \ref{simple_ex} identifies a fundamentally new advantage of team-based incentive pay.  Economists have traditionally justified such schemes by highlighting their role in encouraging cooperation (\cite{Itoh_1991}) or discouraging sabotage (\cite{Lazear_1989}). These channels are explicitly ruled out in the example and in the model studied in the paper. Instead, the advantage of team-based incentive pay is that it allows a profit-maximizing principal to flexibly extract rent. Specifically, the principal reduces expected wage payments when agents take less productive actions than those ``targeted" by her contract --- the situation that matters under worst-case evaluation of contracts.\begin{footnote}{The assumption that agents are identical, i.e., possess a common action set, does introduce a type of perfect correlation in the environment. However, this is a different type of correlation than what has been studied in the literature. Specifically, if performances are perfectly correlated conditional on agents' actions, e.g., there is a common demand shock, then relative performance evaluation, rather than joint performance evaluation, is optimal. Moreover, Section \ref{common} shows that the result in Example \ref{simple_ex} holds even when the common action set assumption is dropped.}\end{footnote} 
   
The baseline model of Section \ref{model} generalizes the example to the setting in which the agents have any number of known and unknown actions. This extension is of particular interest because, with many unknown actions, joint performance evaluation is vulnerable to a pernicious free-riding problem. It is shown that the worst-case payoff for the principal is obtained in the limit of an $n$-sequence of dominance solvable games with $n$ unknown actions. In each game in this sequence, agents ``undercut" each other as dominated strategies are eliminated, taking progressively less costly and less productive actions. Nevertheless, Theorem \ref{JPE_opt} establishes that any worst-case optimal contract is a nonaffine JPE contract.

The proof, outlined in Section \ref{opt_contract}, proceeds by successively improving upon suboptimal contracts. First, it is shown that any contract rewarding failure with strictly positive wages can be improved upon by a contract that does not reward failure (Lemma \ref{sub_linear} and Lemma \ref{zero}). The remaining contracts exhibit either relative performance evaluation (RPE), independent performance evaluation (IPE), or nonaffine joint performance evaluation (JPE). Second, a ranking among optimal contracts within each class is established. It is shown that there does not exist an RPE contract that yields the principal a strictly larger payoff than the optimal IPE contract (Lemma \ref{RPEvIPE}), but that there \textit{does} exist such a JPE contract (Lemma \ref{JPEvIPE}). Establishing this result involves identifying and solving a differential equation characterizing worst-case best-response dynamics over all possible supermodular games induced by the contract (Lemma \ref{JPE_worst}), a technical result that may be of independent interest. 

While I focus on a simple model to isolate the main ideas, I show in Sections \ref{nodiscrim}, \ref{common}, \ref{preferred}, and \ref{otherextensions} that the techniques in the paper can be used to establish principles of worst-case optimal contracts in a variety of extensions. Notably, in a general model with $n \geq 2$ agents and a compact set of output levels, I prove that any worst-case optimal contract is nonaffine and cannot be an IPE. In fact, there always exists a JPE that strictly outperforms the best IPE (Theorem \ref{multiple_thm}). Finally, in Section \ref{finalremarks}, I discuss two applications. First, I observe that the result provides a micro-foundation for empirical evidence documenting firm preferences for nonlinear joint performance evaluation, such as team bonuses, in settings in which production technologies are independent (\cite{Herriesetal_2003}). Second, I observe that, under an appropriate interpretation of the model, it offers an argument for financial diversification that is conceptually distinct from the seminal work of \cite{Diamond_1984}.\begin{footnote}{I thank Lucy White and Huseyin Yildirm for bringing this literature to my attention.}\end{footnote}

 \subsection{Related Literature}\label{relatedlit}
 This paper makes three main contributions to the literature. First, it establishes a fundamentally new justification for joint performance evaluation. In the Bayesian contracting paradigm, the Informativeness Principle prescribes independent performance evaluation whenever one agent's performance is statistically uninformative of another's action. Hence, if the set of actions available to a team of agents is common knowledge, then it is impossible to improve upon independent performance evaluation. To justify incentive schemes commonly used in practice,  such as relative performance evaluation and joint performance evaluation, the literature has instead introduced productive and/or informational linkages between agents.\footnote{In the absence of productive interaction, joint performance evaluation may be optimal if agents are affected by a common, negatively correlated productivity shock (\cite{Fleckinger_2012}). In the absence of a common shock, joint performance evaluation may be optimal if efforts are complements in production (\cite{AlchianDemsetz_1972}), if it induces help between agents (\cite{Itoh_1991}) or, alternatively, if it discourages sabotage (\cite{Lazear_1989}). Finally, joint performance evaluation may be optimal if agents are engaged in repeated production and it allows for more effective peer sanctioning (\cite{CheYoo_2001}).} Specifically, one agent's action either has a direct effect on another's performance and/or there is correlation in performances conditional on an action profile. The model studied in this paper explicitly rules out these channels. It is worth pointing out, however, the literature has not systematically studied the problem of optimal (incomplete) output-contingent contracts under Bayesian uncertainty about the agents' production technology. I consider it in detail in Section \ref{discuss} and Online Appendix \ref{app_incompletebayesian}. There, I provide an alternative, Bayesian rationalization of joint performance evaluation.

Second, it is the first to conduct a formal analysis of a principal-many agents model in which the principal has bounded, non-quantifiable uncertainty about the agents' production technology.\footnote{Related work not discussed here include the papers of \cite{HurwiczShapiro_BJE1978}, \cite{Garrett_2014}, and \cite{Frankel_2014}, and \cite{Rosenthal_2020}, who consider contracting with unknown preferences; \cite{MarkuOcampoDiaz_2019}, who consider a robust common agency problem; and \cite{Chassang_2013}, who studies the robust performance guarantees of a different class of calibrated contracts than those considered here in a dynamic agency problem. At the intersection of computer science and economics, see also \cite{Dutting_2020}, who study near-optimal contracts in principal-agent relationships and \cite{Babaioffetal_2012} who study how the agents' production technology affects whom the principal contracts with, as well as the principal's loss of profits due to moral hazard.} The pioneering work of \cite{Carroll_2015} considers a principal-single agent model in which the principal has non-quantifiable uncertainty about the actions available to the agent. His main result is that there exists a worst-case optimal contract that is linear in individual output. The model and analysis in this paper enrich that of \cite{Carroll_2015} by introducing a seemingly irrelevant agent and showing that multiple agents lead to the optimality of joint incentive schemes.\begin{footnote}{Building upon \cite{Carroll_2015}'s single-agent model, \cite{Antic_2014} imposes bounds on the principal's uncertainty over the productivity of unknown actions (see also Section 3.1 of \cite{Carroll_2015}, which studies lower bounds on costs). In contrast, the model studied here places no restrictions on the technology available to each agent in isolation beyond those of \cite{Carroll_2015}. Instead, the restrictions concern the relationship between the agents. }\end{footnote} 

 \cite{DaiToikka_2018} extend the analysis of \cite{Carroll_2015} to multi-agent settings, but consider a model in which the principal deems \textit{any} game the agents might be playing plausible. In this setting, they find that linear contracts are worst-case optimal. This result is driven by the finding that any contract that induces competition between agents is non-robust to prisoners' dilemma-type games in which one agent's action can directly influence the productivity of another, leading the principal to a worst-case payoff of zero. In contrast to \cite{DaiToikka_2018}, I consider a setting in which the principal \textit{knows} that success is independently distributed across agents. This has the immediate effect of ruling out such games and ensuring that linear contracts are suboptimal. It also necessitates new techniques to analyze the principal's worst-case payoffs.\footnote{For instance, the worst-case payoff of the principal at the optimal contract is achieved by a sequence of games in which the number of actions grows to infinity, rather than one additional action for each agent as in \cite{DaiToikka_2018}.} Despite these differences, the results of this paper complement  \cite{DaiToikka_2018} in terms of their management implications. Agents in  \cite{DaiToikka_2018}'s model are a ``real team" in the sense that they work together to produce value for the principal, while agents in the model of this paper are best thought of as ``co-actors" given the assumption of technological independence (\cite{Hackman_2002}). Yet, in either case, joint performance evaluation is optimal. What changes is the particular form of the optimal joint performance evaluation contract --- in the case of a real team, optimal compensation is linear in the value the team generates for the principal, while in the case of co-acting agents it involves nonlinear bonus payments that reward agents when all succeed.

Third, this paper contributes to the literature on supermodular implementation (\cite{Chen_2002}, \cite{Mathevet_2010}, \cite{MathevetHealy_2012}) by presenting an environment in which a supermodular mechanism (a mechanism inducing a supermodular game between agents) emerges as optimal due to robustness considerations instead of restrictions on the set of feasible mechanisms. Equilibria of supermodular games possess desirable theoretical properties: they can be found by iterated elimination of strictly dominated strategies, and are also the limit points of adaptive and sophisticated learning dynamics (\cite{MilgromRoberts1990}, \cite{MilgromRoberts_1991}). In addition, a collection of experimental papers have shown that laboratory subjects converge to equilibrium faster in supermodular games than in other classes of games (see, for instance, \cite{ChenGazzale_2004}, \cite{Healy_2006}, and \cite{VanEssen_2012}). These benefits of joint performance evaluation are not captured formally in the model of this paper, but might further justify their use in practice.

\needspace{2 \baselineskip}

\section{Model}\label{model}

\subsection{Environment}

A risk-neutral principal writes a contract for two risk-neutral agents, indexed by $i=1,2$. Each agent $i$ chooses an unobservable action, $a_i$, from a common, finite set $A \subset \mathbb{R}_+ \times [0,1]$ to produce individual output $y_i \in \{0,1\}$, where $y_i=1$ indicates ``success" and $y_i=0$ indicates ``failure". Each action $a_i$ is identified by its cost, $c(a_i) \in \mathbb{R}_+$, and the probability with which it results in success, $p(a_i) \in [0,1]$.  There are no informational linkages across agents:
\[ Pr(y_i, y_j |a_i, a_j)= Pr(y_i| a_i, a_j) Pr(y_j| a_i, a_j). \]
There are no productive linkages across agents:
\[ Pr(y_i| a_i, a_j)= Pr(y_i| a_i)= \begin{cases} p(a_i) & \text{if $y_i=1$} \\ 1-p(a_i) & \text{if $y_i=0$} \end{cases}. \] 

\subsection{Contracts}
A \textbf{contract} is a quadruple of non-negative wages, $w:=(w_{11},w_{10},w_{01}, w_{00}) \in \mathbb{R}^4_+$,
where the first index of each wage indicates an agent's own success or failure and the second indicates the success or failure of the other agent.\begin{footnote}{ I impose the assumption that contracts are symmetric, i.e., nondiscriminatory, throughout, postponing a discussion of asymmetric contracts to Section \ref{otherextensions} and Online Appendix \ref{discrimination}. I also discuss a precise sense in which contracts are incomplete in Section \ref{discuss}.}\end{footnote} It will be useful to classify contracts according to the following typology of \cite{CheYoo_2001}.\begin{footnote}{While this typology is non-exhaustive (for instance, when $w_{11}>w_{10}$ and $w_{01}<w_{00}$ there is JPE ``at the top" and RPE ``at the bottom"), I will show later that it is without loss of generality to consider contracts for which $w_{01}=w_{00}=0$ (Lemma \ref{zero}). Within this class of contracts, it \textit{is} exhaustive. }\end{footnote}

\begin{definition}[Performance Evaluations]
A contract $w$ is
\begin{itemize}
\item an \textbf{independent performance evaluation (IPE)} if $(w_{11},w_{01})=(w_{10}, w_{00})$;
\item a \textbf{relative performance evaluation (RPE)} if $(w_{11},w_{01})<(w_{10}, w_{00})$;
\item and a \textbf{joint performance evaluation (JPE)} if $(w_{11},w_{01})>(w_{10}, w_{00})$,
\end{itemize}
where $>$ and $<$ indicate strict inequality in at least one component and weak in both.
\end{definition}

\noindent It will also be useful to delineate which contracts are affine.

\begin{definition}
A contract is \textbf{affine} if \[w_{y_i y_j}= \alpha_0+ \alpha_i y_i+ \alpha_j y_j \quad \text{for $\alpha_0, \alpha_i, \alpha_j \geq 0$,}\] and \textbf{nonaffine} otherwise. 
\end{definition}
 \noindent Two remarks are in order. First, notice that an IPE contract is an affine contract with $\alpha_j=0$. Second, notice that a \textbf{linear} contract in the sense of \cite{DaiToikka_2018} is an affine JPE contract with $\alpha_0=0$ and $\alpha_i=\alpha_j$. 

Agent $i$'s ex post payoff given a contract $w$, action profile $(a_i, a_j)$, and realization $(y_i,y_j)$ is \[w_{y_i y_j}-c(a_i),\] while her expected payoff is
\begin{equation}
U_i(a_i, a_j; w) :=  \sum_{y_i } \sum_{y_j} Pr(y_i, y_j | a_i, a_j) w_{y_i y_j}-c(a_i).
\notag
\end{equation} 
Let $\Gamma(w,A)$ denote the normal form game induced by the contract $w$ and $\mathcal{E}(w,A)$ denote its (non-empty) set of mixed strategy Nash equilibria.

\subsection{Principal's Problem}\label{principal_problem}

The principal's ex post payoff given a contract $w$ and realization $(y_1,y_2)$ is \[y_1+y_2-w_{y_1 y_2}-w_{y_2 y_1},\] 
while her expected payoff is
\[ V(w,A) := \underset{ \sigma \in \mathcal{E}(w,A)}{\max} E_{\sigma}[ y_1+y_2-w_{y_1 y_2}-w_{y_2 y_1} ].\] Notice that the principal can select her preferred Nash equilibrium in case of multiplicity. This minimizes the distance between the model studied here and the literature discussed in Section \ref{relatedlit}. However, many results persist under weaker selection assumptions as discussed in Section \ref{otherextensions}.

When the principal writes a contract for the agents, she has limited knowledge about the game the agents play. In particular, she knows only a non-empty subset of actions available to them $A^0 \subseteq A$. In the face of her uncertainty, the principal evaluates each contract on the basis of its performance across all finite supersets of her knowledge. The \textbf{worst-case payoff} she receives from a contract $w$ is thus given by
\[ V(w):= \underset{A \supseteq A^0}{\inf}~V(w,A). \] 
The principal's problem is to identify a contract $w^*$ for which \[V(w^*)= \underset{w}{\sup}~V(w).\] Call such a contract a \textbf{worst-case optimal contract}.

\section{Analysis}\label{opt_contract}

To rule out uninteresting cases, I make the following assumption about $A^0$ in the subsequent analysis.

\begin{assumption}\label{assumption}
The known action set $A^0$ has the following properties:
\begin{enumerate}
\item (\textit{Non-Triviality}) There exists an action $a_0 \in A^0$ such that \[p(a_0)-c(a_0)>0.\]
\item (\textit{Known Actions are Costly}) If $a_0 \in A^0$, then $c(a_0)>0$. 
\end{enumerate}

\end{assumption}

The first assumption ensures that the principal can possibly obtain a strictly positive worst-case payoff from contracting with the agents. The second ensures that the principal's supremum payoff is never approached by a sequence of contracts converging to the contract that always pays zero.\footnote{While the first assumption is necessary for the main result, the second is not. In particular, as long as the principal does not ``target" any zero-cost action, the result goes through. I maintain this assumption due to its ease of interpretation and because it eliminates some nuisance cases in the proof. }

\subsection{Main Result}

The main result follows below.

\begin{theorem}\label{JPE_opt}
Any worst-case optimal contract is a nonaffine JPE. There exists a worst-case optimal contract.
\end{theorem}

The key intuition behind the result is that by judiciously calibrating a JPE to a benchmark IPE, any efficiency losses such contracts generate can be made approximately the same as those of the benchmark contract. Thus, the reduction in expected wage payments the principal obtains when agents take less productive actions causes JPE to outperform the benchmark contract. Of course, to show that only nonaffine JPE can be worst-case optimal, I must also prove strict suboptimality of contracts other than IPE, including those that exhibit RPE.

The proof has five steps. First, I show that no affine contract can outperform the best IPE (Lemma \ref{sub_linear}) and that, more generally, any contract rewarding failure with strictly positive wages can be improved by a contract $w$ for which $w_{01}=w_{00}=0$ or which yields a worst-case payoff smaller than that of the best IPE  (Lemma \ref{zero}). Consequently, to identify a worst-case optimal contract, it suffices to consider those which are either RPE ($w_{11}<w_{10}$), IPE ($w_{11}=w_{10}$), or JPE ($w_{11}>w_{10}$). Second, I show that there does not exist an RPE that yields the principal a strictly larger payoff than the best IPE (Lemma \ref{RPEvIPE}). Third, I compute the principal's worst-case payoff given any JPE (Lemma \ref{JPE_worst}). Fourth, I show that there exists a (calibrated) JPE that yields a strictly higher payoff than the best IPE (Lemma \ref{JPEvIPE}). Fifth, I establish existence of a worst-case optimal nonaffine JPE and that no other class of contracts can be optimal. The remainder of this section outlines these steps.

\subsection{Preliminaries: Supermodular Games}\label{superprelim}

The proof will utilize some results from the theory of supermodular games, which I review now. Equip any action set $A$ with the total order $\succeq$: $a_i \succeq a_j$ if either $p(a_i) > p(a_j)$, or $p(a_i)=p(a_j)$ and $c(a_i) \leq c(a_j)$. In words, $a_i$ is higher than $a_j$ if $a_i$ results in success with a higher probability or if it results in success with the same probability, but at a lower cost. Then, $(A,\succeq)$ is a complete lattice. A supermodular game may thus be defined as follows.

\begin{definition}[Supermodular Games]\label{supermodular} The game $\Gamma(w,A)$ is \textbf{supermodular} if $U_i$ exhibits increasing differences: $a'_i \succeq a_i$ and $a'_j \succeq a_j$ implies
\[ U_i(a'_i, a'_j; w)-U_i(a_i,a'_j; w) \geq U_i(a'_i, a_j; w)-U_i(a_i,a_j; w).\] It is \textbf{submodular} if $U_i$ exhibits decreasing differences: $a'_i \succeq a_i$ and $a'_j \succeq a_j$ implies
\[ U_i(a'_i, a'_j; w)-U_i(a_i,a'_j; w) \leq U_i(a'_i, a_j; w)-U_i(a_i,a_j; w).\]
\end{definition}

The important property of supermodular games that I exploit is that best-response dynamics converge to their maximal and minimal equilibria. In particular, let $a_{\max}$ and $a_{\min}$ denote the maximal and minimal elements of $A$, and $\overline{BR}: A \rightarrow A$ and $\underline{BR}: A \rightarrow A$ denote the maximal and minimal best-response functions for the agents. Then, the following Lemma holds.

\begin{lemma}[\cite{Vives_1990}, \cite{MilgromRoberts1990}]\label{lemma_super}
Suppose $\bar{a}$ ($\underline{a}$) is the limit found by iterating $\overline{BR}$ ($\underline{BR}$) starting from $a_{\max}$ ($a_{\min}$). If $\Gamma(w,A)$ is supermodular, then it has a maximal Nash equilibrium $(\bar{a}, \bar{a})$ and a minimal Nash equilibrium $(\underline{a},\underline{a})$; any other equilibrium $(a_i,a_j)$ must satisfy $\bar{a} \succeq a_i \succeq \underline{a}$ and $\bar{a} \succeq a_j \succeq \underline{a}$.
\end{lemma}

A similar property holds for two-player submodular games. Define the mapping
\begin{equation}
\begin{aligned}
\widetilde{BR}: A \times A &\rightarrow A \times A\\
(a_i, a_j) &\mapsto (\overline{BR}(a_j), \underline{BR}(a_i)).
\end{aligned} \notag
\end{equation}
Then, the following Lemma holds.

\begin{lemma}[\cite{Vives_1990}, \cite{MilgromRoberts1990}]\label{lemma_sub}
Suppose $(\bar{a}, \underline{a})$ is the limit found by iterating $\widetilde{BR}$ starting from the action profile $(a_{\max}, a_{\min})$. If $\Gamma(w,A)$ is submodular, then both $(\bar{a}, \underline{a})$ and $(\underline{a},\bar{a})$ are Nash equilibria and any other Nash equilibrium action must be smaller than $\bar{a}$ and larger than $\underline{a}$. 
\end{lemma}

\subsection{Proof of Main Result}

Say that a contract $w$ is \textbf{eligible} if $V(w)>0$.\footnote{This definition implies eligibility in the sense of \cite{Carroll_2015}, who requires that, in addition, $V(w)$ yields a higher worst-case payoff than the contract paying zero wages for all pairs $(y_i, y_j)$. By the assumption of costly known actions, such a contract yields the principal a worst-case payoff of zero.} It is without loss of generality to restrict attention to eligible contracts; \cite{Carroll_2015} already identifies that
\[ V^*_{IPE}:= \underset{\text{$w$: $w$ is an IPE}}{\sup}~ V(w)=2 \underset{w \in [0,1], a_0 \in A^0}{\max}~\left[(p(a_0)-\frac{c(a_0)}{w})(1-w)\right]>0\]  by an argument that generalizes the one sketched in Example \ref{simple_ex}. Hence, any contract $w$ for which $V(w) \leq 0$ cannot be worst-case optimal. 

\subsubsection{Suboptimality of Affine and Related Contracts}

I first provide a simple proof that no affine contract can outperform the best IPE. 

\begin{lemma}\label{sub_linear}
For any affine contract $w$, $V(w) \leq V^*_{IPE}$.
\end{lemma}

\begin{proof}
Suppose $w$ is an affine contract with parameters $\alpha_0, \alpha_i, \alpha_j \geq 0$. Consider an IPE contract $w'$ with parameters $\alpha'_0=\alpha'_j=0$. I claim that this contract weakly increases the principal’s worst-case payoff. First, observe that, for any $A \supseteq A_0$, the incentives of the agents are unchanged; a constant shift in an agent’s payoff holding fixed the action of the other does not affect her optimal choice of action. Hence, $\sigma \in \mathcal{E}(w,A)$ if and only if $\sigma \in \mathcal{E}(w',A)$. Second, observe that, for any equilibrium $\sigma \in \mathcal{E}(w,A)=\mathcal{E}(w',A)$, the principal's expected payoff under $w'$ is weakly larger than under $w$; her expected wage payments decrease and each agent's productivity is unchanged. Hence, $V(w',A) \geq V(w,A)$ for any $A \supseteq A^0$. It follows that
\[ V(w)=  \underset{A \supseteq A^0}{\inf}~V(w,A) \leq \underset{A \supseteq A^0}{\inf}~V(w',A)= V(w') \leq V^*_{IPE}.\]\end{proof}

More generally, any eligible contract $w$ with $w_{00}>0$ or $w_{01}>0$ can be improved upon by another contract $w'$ with $w'_{00}=w_{01}=0$ or, alternatively, cannot yield a payoff higher than $V^*_{IPE}$.

\begin{lemma}[Suboptimality of Positive Wages for Failure]\label{zero}
For any eligible contract $w$ with $w_{00}>0$ or $w_{01}>0$, there either exists a contract $w'$ with $w'_{01}=w'_{00}=0$ and $V(w') \geq V(w)$, or $V^*_{IPE} \geq V(w)$.
\end{lemma}
\begin{proof}
See Appendix \ref{proof_zero}.
\end{proof}

\noindent Though the result is familiar, the proof is surprisingly nontrivial. Specifically, while the ``shifting" argument used in the proof of Lemma \ref{sub_linear} rules out many contracts, there are two cases that require different arguments. First, when $w_{11}>0$ and $w_{00}>0$ (with $w_{01}=w_{00}=0$), I exploit supermodularity of the payoff function and a comparative statics result of \cite{MilgromRoberts1990} to argue that the probability of success under any equilibrium action decreases in $w_{00}$. Second, when $w_{10}>0$ and $w_{01}>0$ (with $w_{11}=w_{00}=0$), I must rule out asymmetric and mixed equilibria that might be beneficial for the principal.  I therefore encourage the interested reader to review it only upon reading the rest of Section \ref{opt_contract}.

An immediate corollary of Lemma \ref{zero} is that to find a worst-case optimal contract it suffices to consider nonaffine JPE satisfying $w_{11}>w_{10}$, IPE satisfying $w_{11}=w_{10}$, and RPE satisfying $w_{11}<w_{10}$. I next establish a ranking among the classes of JPE, IPE, and RPE contracts, exploiting the following observation.

\begin{observation}\label{obs_1}
If $w$ is an RPE for which $w_{00}=w_{01}=0$ and $A \supseteq A^0$, then $\Gamma(w,A)$ is a submodular game. If $w$ is a JPE for which $w_{00}=w_{01}=0$ and $A \supseteq A^0$, then $\Gamma(w,A)$ is a supermodular game.
\end{observation}

\subsubsection{RPE Cannot Outperform IPE}\label{RPEvIPE}

I now establish that no RPE can yield a higher payoff than the best IPE.

\begin{lemma}[IPE Outperforms RPE] \label{RPEvIPE}
No RPE with $w_{01}=w_{00}=0$ can yield the principal a higher worst-case payoff than $V^*_{IPE}$.
\end{lemma}

\begin{proof}
See Appendix \ref{IPE>RPE}.
\end{proof}

I sketch the proof for the case in which there is a single known action, i.e., $A^0:=\{a_0\}$. Suppose each agent has available a single additional zero-cost action $a^*$ that results in success with probability $p(a^*)<p(a_0)$. Then, $a^*$ is a strict best response to $a^*$ if and only if
\[ \underbrace{p(a^*) \left(p(a^*) w_{11}+ (1-p(a^*)) w_{10} \right)}_\text{Payoff $a^*$ against $a^*$}> \underbrace{p(a_0) \left(p(a^*) w_{11}+(1-p(a^*)) w_{10} \right)-c(a_0)}_\text{Payoff $a_0$ against $a^*$}   \]
\begin{equation} \iff p(a^*) > p(a_0)-\frac{c(a_0)}{p(a^*) w_{11}+(1-p(a^*)) w_{10}}. \label{eq_RPE} \notag \end{equation} 
This condition also ensures that $a^*$ is a strictly dominant strategy because any RPE induces a submodular game between the agents. Intuitively, if $a^*$ is a strict best response to $a^*$, which is less productive than $a_0$, then it must also be a strict best response to $a_0$; the marginal benefit of shirking against a more productive action is higher (because $w_{10}>w_{11}$). The principal's payoff as $p(a^*)$ approaches the value at which the incentive constraint binds is therefore
\[ 2 \underbrace{(p(a_0)-\frac{c(a_0)}{p(a^*) w_{11}+(1-p(a^*)) w_{10}})}_\text{Probability Success} \times \underbrace{\left[ 1-(p(a^*) w_{11}+(1-p(a^*)) w_{10}) \right]}_\text{Conditional Expected Surplus}.\]
Letting $\hat{w}:=p(a^*) w_{11}+(1-p(a^*)) (1-w_{10})$, it is immediate that she can do no better than $V^*_{IPE}$: 
\[2 (p(a_0)-\frac{c(a_0)}{\hat{w}})(1- \hat{w}) \leq  2 \underset{w \in [0,1]}{\max}~\left[(p(a_0)-\frac{c(a_0)}{w})(1-w)\right] = V^*_{IPE}.\]
The proof for general known action sets uses a fixed-point theorem to identify the existence of a worst-case equilibrium $(a^*,a^*)$.

\subsubsection{JPE Worst-Case Payoffs}\label{analysis_JPE}
 Within the class of contracts setting $w_{00}=w_{01}=0$, the only contracts left to consider are nonaffine JPE for which $w_{11}>w_{10}$. (Notice that such contracts can be re-written in the form described in Example \ref{simple_ex} by defining $w_0:=w_{10}$ and $b:=w_{11}-w_{10}$.) Lemma \ref{JPE_worst} states the principal's worst-case payoff guarantee from any contract of this form.
 
\begin{lemma}[JPE Worst-Case Payoffs] \label{JPE_worst}
Suppose $w$ is a JPE with $w_{00}=w_{01}=0$ and, for each $a_0 \in A^0$, $\hat{p}(\cdot | a_0): [0,\hat{t}(a_0)] \rightarrow [0,p(a_0)]$ is the unique solution to the initial value problem
\begin{equation}
\begin{aligned}
\hat{p}'(t) =& f(\hat{p}(t)):= - \left[ \hat{p}(t) w_{11} +(1-\hat{p}(t)) w_{10}\right]^{-1}~~\text{with} \\
 \hat{p}(0)=& p(a_0),
\end{aligned}
\label{diffeq}
\end{equation}
where $[0,\hat{t}(a_0)] \subseteq  [0,c(a_0)]$ is the largest interval on which $\hat{p}(t)>0$ for all $t \in [0,\hat{t}(a_0))$. Then,
\begin{equation}
V(w) = 2~ \min\{1-w_{11}, \bar{p} \left[ \bar{p} (1-w_{11})+ (1-\bar{p}) (1- w_{10}) \right] \},
\label{worst}
\end{equation}
where \[\bar{p} := \underset{a_0 \in A^0}{\max}~~\hat{p}(\hat{t}(a_0) | a_0).\]
\end{lemma}
\begin{proof}
See Appendix \ref{JPEworst_proof}.
\end{proof}
The principal's worst-case payoff, $V(w)$, is two times the minimum of two terms. The first term is the principal's payoff from each agent when the worst-case action set induces a game between the agents in which there is a unique equilibrium in which both succeed with probability one. The second term is the principal's payoff when the worst-case action set induces a game between the agents in which, in the maximal equilibrium induced by the contract, each succeeds with a probability $\bar{p}$ as low as possible. (Both are required because, for high enough $w_{11}$, the principal may prefer the ``shirking equilibrium".) Rather than outline the entire proof of Lemma \ref{JPE_worst}, I instead describe the sequence of games that leads to the worst-case distribution $\bar{p}$, focusing on why the ``one unknown action" construction of Example \ref{simple_ex} is insufficient.

\paragraph{The Worst-Case Sequence of Games.}

For simplicity, suppose there is a single known action $a_0$ with success probability $p(a_0)=1$ and cost $c(a_0)=\frac{1}{4}$. The optimal IPE puts $w^*=w_{11}=w_{10}=\frac{1}{2}$. Given $w^*$, the worst-case success probability approaches \[p(a_0)-\frac{c(a_0)}{w^*}=\frac{1}{2}.\] 
Now, suppose I reduce $w_{10}$ to zero, but keep all other wages the same. This contract is (trivially) calibrated to $w^*$ and the known action $a_0$:
\[  \underbrace{p(a_0)}_{=1} \underbrace{w_{11}}_{=\frac{1}{2}}+ \underbrace{(1-p(a_0))}_{=0} w_{10} =\underbrace{w^*}_{=\frac{1}{2}}.\]
So, according to the analysis in Example \ref{simple_ex}, there is ostensibly \textit{no} efficiency loss generated by this modification. 

In particular, if I consider only the class of games with action sets of the form $A^1:=A^0 \cup \{a^1_1\}$, for some action $a^1_1$ with success probability $p(a^1_1)<p(a_0)$, then the worst case for the principal occurs as $p(a^1_1)$ approaches the value at which the best-response condition binds:  \[p(a^1_1) \left[ p(a_0) w_{11}+(1-p(a_0)) w_{10} \right] -c(a^1_1)=p(a_0) \left[ p(a_0) w_{11}+(1-p(a_0)) w_{10} \right] -c(a_0) \] \[\iff p(a^1_1)=p(a_0)-\frac{c(a_0)-c(a^1_1)}{p(a_0) w_{11}+(1-p(a_0)) w_{10} }\geq \frac{1}{2}.\] See Figure \ref{fig_onestep} for a geometric representation of this argument.

\begin{figure}[t]
\centering
\includegraphics[scale=0.25]{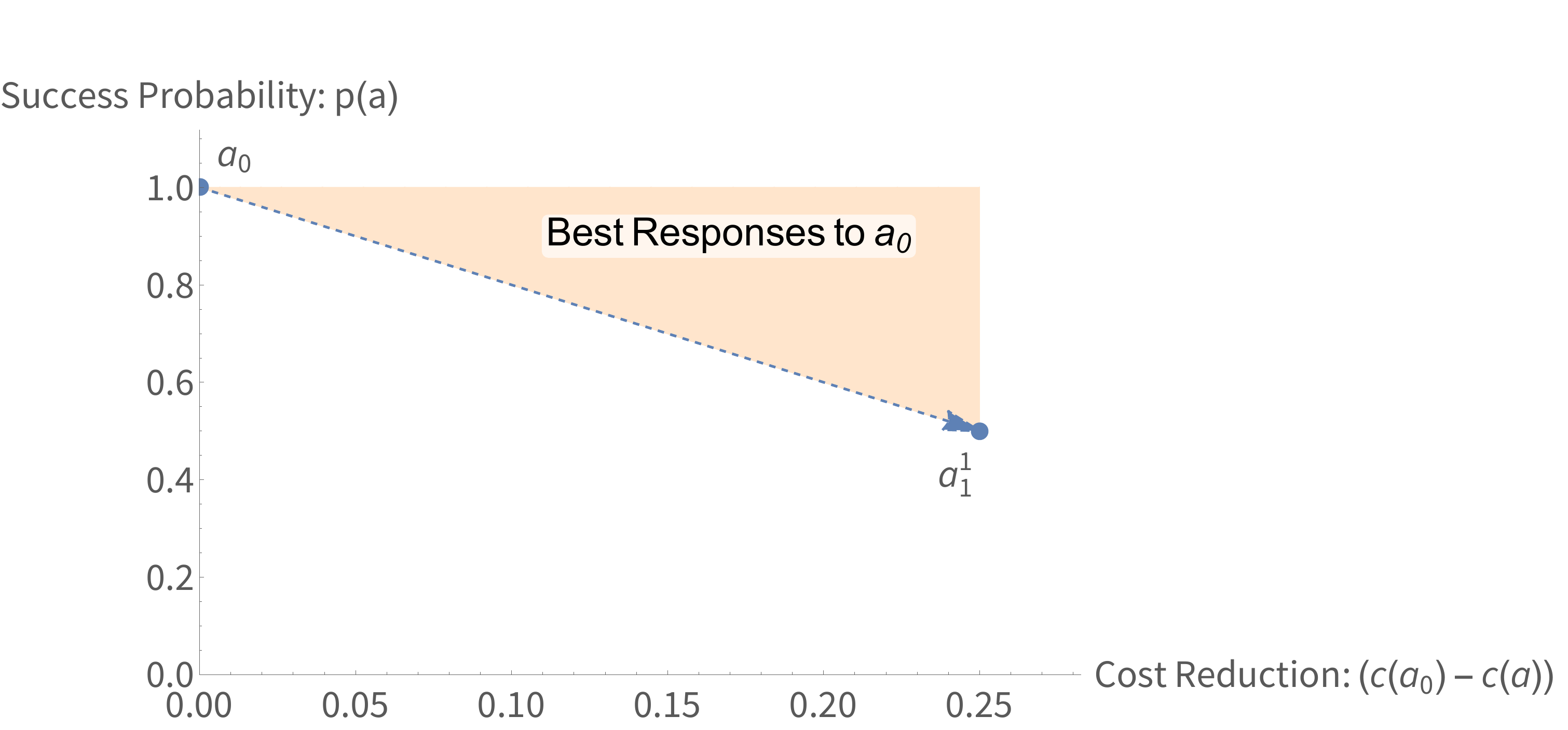}
\caption{$A^1$ best-response path.}\label{fig_onestep}
\caption{This figure depicts the best-response response path starting from the known (maximal) action $a_0$. The dashed line may be interpreted as an \textit{indifference curve} with slope $-1/(p(a_0) w_{11}+(1-p(a_0)) w_{10})$ and intercept $p(a_0)$: each action on the line, $a$, is identified by its cost relative to $c(a_0)$, $x=c(a_0)-c(a)$, and its success probability, $y=p(a)$. Since the slope of the indifference curve is negative, the maximal reduction in success probability occurs when the cost reduction is as large as possible, i.e., when $c(a^1_1)=0$ so that $x=\frac{1}{4}$. As $p(a^1_1) \downarrow \frac{1}{2}$, the worst-case probability is achieved.}
\end{figure}

But what if there are two unknown actions? Consider the action set $A^2:=A^0 \cup \{a^2_1, a^2_2\}$, where $a^2_1$ has a positive cost of $c(a^2_1)= \frac{c(a_0)}{2}=\frac{1}{8}$ and $c(a^2_2)=0$. A simple calculation shows that for $a^2_1$ to be a strict best-response to $a_0$, it must be the case that
 \[p(a^2_1) \left[ p(a_0) w_{11}+(1-p(a_0)) w_{10} \right] -c(a^2_1)>p(a_0) \left[ p(a_0) w_{11}+(1-p(a_0)) w_{10} \right] -c(a_0) \] \[\iff p(a^2_1)>p(a_0)-\frac{c(a_0)-c(a^2_1)}{p(a_0) w_{11}+(1-p(a_0)) w_{10}}= \frac{3}{4}.\] Furthermore, for $a^2_2$ to be a best-response to $a^2_1$, it must be the case that 
\[p(a^2_2)>p(a^2_1)-\frac{c(a^2_1)-c(a^2_2)}{p(a^2_1) w_{11}+(1-p(a^2_1)) w_{10}}=p(a^2_1)-\frac{1}{4 p(a^2_1)}.\]
If $p(a^2_1)$ is close to $\frac{3}{4}$ and $p(a^2_2)$ is close to $p(a^2_1)-1/(4 p(a^2_1))$, then, in addition, $a^2_1$ is the unique best-response to $a_0$ and $a^2_2$ is the unique best-response to $a^2_1$. Hence, best-response dynamics under the operator $\overline{BR}$ converge to $(a^2_1, a^2_1)$ starting from the maximal action in $A^2$, $a_0$. Since $\Gamma(w,A^2)$ is a supermodular game  (Observation \ref{obs_1}), Lemma \ref{lemma_super} thus implies that $(a^2_1, a^2_1)$ is the maximal Nash equilibrium.\begin{footnote}{ A similar argument shows that  best-response dynamics under the operator $\underline{BR}$ converge to $(a^2_1, a^2_1)$ starting from the minimal action in $A^2$, $a^2_1$. Hence, it is in fact the unique Nash equilibrium.}\end{footnote} In it, each agent's success probability can be made arbitrarily close to \[\frac{3}{4}-\frac{1}{4 \times \frac{3}{4}}=\frac{5}{12}<\frac{1}{2}.\]
See Figure \ref{fig_twosteps}, which continues the geometric argument.

\begin{figure}[t]
\centering
\includegraphics[scale=0.5]{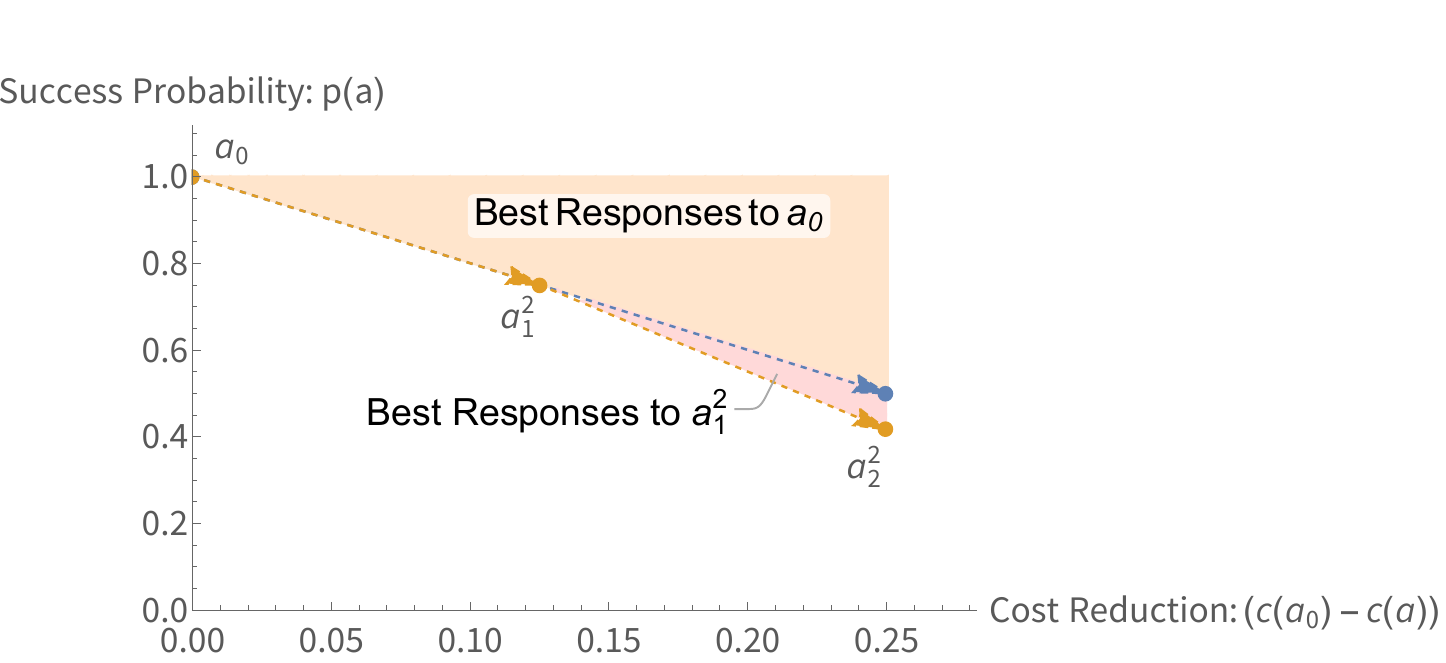}
\caption{$A^2$ best-response path.}\label{fig_twosteps}
\end{figure}

I now generalize this construction to drive the equilibrium probabilities of success even lower. Let $A^n:= A^0 \cup \{a^n_1,...,a^n_n\}$ be an action set with $c(a^n_k)=(n-k)\frac{c(a_0)}{n}$, so that costs are evenly distributed on a grid between zero and $c(a_0)$. For each $k=1,...,n$, choose $p(a_k)$ so that $a_k$ is a best-response to $a_{k-1}$, i.e. set
\begin{equation} p(a_k)= p(a_{k-1})- \epsilon(n) \left[p(a_{k-1}) w_{11}+(1-p(a_{k-1})) w_{10}\right]^{-1}+\rho(n), \label{difference}\tag{E}
 \end{equation}
where $\epsilon(n):= \frac{c(a_0)}{n}$ and $\rho(n)>0$.\footnote{To see why this is an equivalent condition, multiply both sides of the equation by $p(a_{k-1}) w_{11}+(1-p(a_{k-1})) w_{10}$.} For $\rho(n)$ small, $a_k$ is a maximal best-response to $a_{k-1}$ for all $k$. It follows that the maximal Nash equilibrium of $\Gamma(w,A^n)$ is $(a^n_n, a^n_n)$, found again by iterating best-responses. Hence, the per-agent probability of success can be reduced to $p(a^n_n)$.

\begin{figure}[t]
\centering
\includegraphics[scale=0.25]{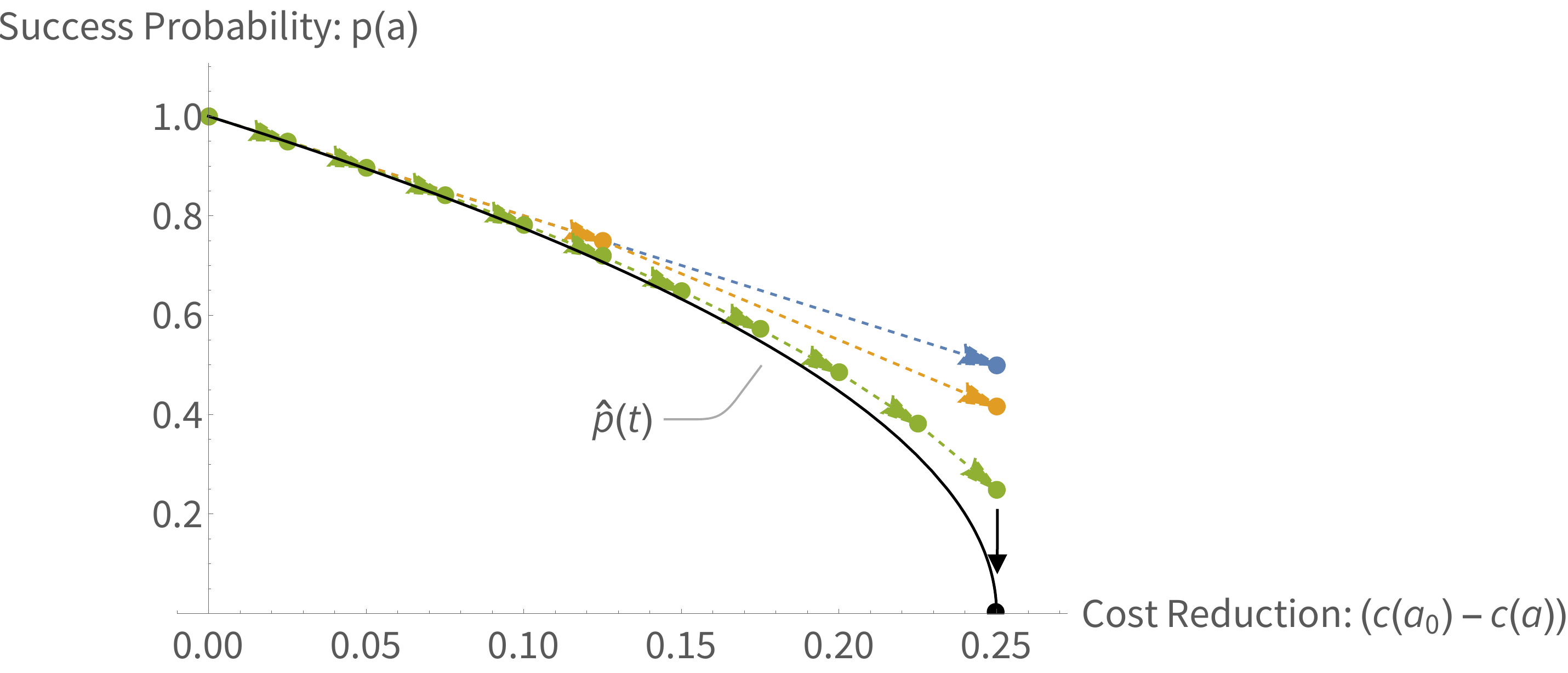}
\caption{$p(a^n_n)$ as $n \rightarrow \infty$.} \label{limit}
\end{figure}

It turns out that $\bar{p}$ is the limit of $p(a^n_n)$ as $n \rightarrow \infty$. To prove this, I observe that Equation \ref{difference} is an Euler approximation of Equation \ref{diffeq}, where $\frac{c(a_0)}{n}$ is the step size of the approximation and $\rho(n)$ is a ``rounding error". Hence, as $n$ grows large, if the rounding error $\rho(n)$ approaches zero at an appropriately fast rate relative to $\epsilon(n)$, agents' best-response dynamics are well-described by the solution to Equation \ref{diffeq},  $\hat{p}(\cdot | a_0)$, under the interpretation that $t$ is ``cost-reduction relative to $a_0$".\footnote{See, for instance, Theorem 6.3 of \cite{Atkinson_1989} and the proceeding discussion.} In the example considered here, the limit is \[\bar{p}=\hat{p}(\hat{t}(a_0) | a_0)=\hat{p}(c(a_0)|a_0)=\hat{p}(0.25 |a_0)=0,\] as depicted in Figure \ref{limit}. As zero profits are obtained in this limit, the so-constructed JPE cannot outperform the optimal IPE to which it was calibrated.

\subsubsection{Existence of a Calibrated JPE Outperforming IPE}\label{rent}

While I demonstrated in the previous section that not \textit{every} calibrated JPE outperforms the optimal IPE, I prove that there must exist one that does. 

\begin{lemma}[JPE Outperforms IPE] \label{JPEvIPE}
There exists a JPE with $w_{00}=w_{10}=0$ yielding the principal a strictly higher worst-case payoff than $V^*_{IPE}$.
\end{lemma}

\begin{proof}
See Appendix \ref{JPEIPE_proof}.
\end{proof}

I illustrate the proof using the running example with a single known action $a_0$ for which $p(a_0)=1$ and $c(a_0)=\frac{1}{4}$. As previously pointed out, the optimal IPE given this action puts $w^*=w_{11}=w_{10}=\frac{1}{2}$. Now, consider the calibrated JPE setting $w_{10}=\frac{1}{2}-\epsilon$ for $\epsilon>0$ and $w_{11}= \frac{1}{2}$. I show that the calibrated JPE strictly increases the principal's worst-case payoff if $\epsilon$ is sufficiently small. Integration reveals that the solution to the differential equation defining $\bar{p}$ in Lemma \ref{JPE_worst} is \[ \bar{p}(\epsilon):= \frac{\sqrt{\frac{1}{2}(\frac{1}{2}-\epsilon)}-(\frac{1}{2}-\epsilon)}{\epsilon}. \]
By L'Hôpital's rule, as $\epsilon \rightarrow 0^+$, so that the wage scheme constructed approaches the optimal IPE, $\bar{p}(\epsilon)$ approaches $\frac{1}{2}$, the worst-case equilibrium probability of success given the optimal IPE. Differentiating $\bar{p}(\epsilon)$ and taking its limit as $\epsilon \rightarrow 0^+$, I identify a local calibration effect on the worst-case probability of success:
\[\underset{ \epsilon \rightarrow 0^+}{\lim}~~ \bar{p}'(\epsilon)=-\frac{1}{4}.\]
I now compute the local effect of calibration on the principal's \textit{profit} from each agent in the shirking equilibrium.\begin{footnote}{As the principal's profit in the shirking equilibrium at the optimal IPE is strictly lower than in the equilibrium in which both agents succeed with probability one, it suffices to show that the principal benefits from such a decrease to exhibit a strict increase in the principal's payoff.}\end{footnote} For any $\epsilon>0$, the principal's payoff per agent in the shirking equilibrium is
\[ \pi(\epsilon):= \underbrace{\bar{p}(\epsilon)}_\text{Expected Task Value} \times \underbrace{\left[ 1- (\bar{p}(\epsilon) w_{11}+ (1-\bar{p}(\epsilon)) w_{10}) \right]}_\text{Conditional Expected Surplus}. \] Using the product rule and taking limits,
\begin{equation}
\begin{aligned}
\underset{ \epsilon \rightarrow 0^+}{\lim}~ \pi'(\epsilon)&= \underset{ \epsilon \rightarrow 0^+}{\lim} \underbrace{ \bar{p}'(\epsilon) \left( 1- (\bar{p}(\epsilon) w_{11}+ (1-\bar{p}(\epsilon)) w_{10}) \right) }_\text{Efficiency Loss}+  \underbrace{\bar{p}(\epsilon)\left( \frac{p(a_0)-\bar{p}(\epsilon)}{p(a_0)}-\bar{p}'(\epsilon) \frac{\epsilon}{p(a_0)} \right) }_\text{Gain in Conditional Rents} \\
&= -\frac{1}{4} \times \frac{1}{2}+ \frac{1}{4}>0.
\end{aligned}
\notag \end{equation}
This establishes the desired result. 

The economic intuition behind the final calculation is as follows. When increasing $w_{11}$ and decreasing $w_{10}$, there are two effects. The first effect is that the probability of task completion in the worst-case Nash equilibrium is \textit{marginally} decreased, as captured by the term $\underset{ \epsilon \rightarrow 0^+}{\lim}~~ \bar{p}'(\epsilon)$ in the efficiency loss expression. On the other hand, expected wage payments are decreased in \textit{level}. Specifically, if the worst-case expected productivity of both agents is fixed at its value under the optimal IPE, $\underset{ \epsilon \rightarrow 0^+}{\lim}~~ \bar{p}(\epsilon)$, then expected wage payments to each agent conditional on success decrease by an amount
\begin{equation}
    \begin{aligned}
    w^*- \left(\left(\underset{ \epsilon \rightarrow 0^+}{\lim} \bar{p}(\epsilon) \right) w_{11}+\left(1-\left(\underset{ \epsilon \rightarrow 0^+}{\lim}\bar{p}(\epsilon)\right) \right)w_{10} \right)=  \underset{ \epsilon \rightarrow 0^+}{\lim} \left( \frac{p(a_0)-\bar{p}(\epsilon)}{p(a_0)} \right).
    \end{aligned} \notag
\end{equation} 
The final term is the non-vanishing term in the ``gain in conditional rents" expression and corresponds to the percentage decrease in the productivity of the ``worst-case" action relative to the targeted known action $a_0$. The proof establishes that the level decrease in expected wage payments dominates the marginal decrease in expected productivity due to compounding shirking behavior. Put differently, local changes away from the optimal IPE make it difficult for best-response dynamics to generate enough momentum to outweigh the level rent-extraction gain.

\subsubsection{Existence and Uniqueness}\label{laststeps}

To establish existence of a worst-case optimal contract, I simply observe that the search for an optimal JPE with $w_{00}=w_{01}=0$ can be recast as a maximization problem of a continuous function over a compact set. To establish that any worst-case optimal contract must be a JPE with $w_{00}=w_{01}=0$, I need only strengthen the proof of Lemma \ref{zero} to show that any contract $w$ with either $w_{00}>0$ or $w_{01}>0$ is either weakly outperformed by an IPE or RPE, or strictly outperformed by a JPE. I leave these last technical details to Appendix \ref{details_laststeps}, thereby completing the proof of Theorem \ref{JPE_opt}.

\subsection{Optimal Wages}
To conclude the analysis of the baseline model, I observe that optimal values of $w_{11}$ and $w_{10}$ can be found by solving the following maximization problem:
\begin{equation}
	\underset{w_{11}>w_{10} \geq 0}{\max}~~ \min\{1-w_{11}, \bar{p}(w_{11},w_{10}) \left[ \bar{p}(w_{11},w_{10}) (1-w_{11})+ (1-\bar{p}(w_{11},w_{10})) (1- w_{10}) \right] \},
	\notag \end{equation}
where $\bar{p}(w_{11}, w_{10})$ is the solution to the initial value problem in the statement of Lemma \ref{JPE_worst} and the dependence on $w_{11}$ and $w_{10}$ has been made explicit. For any eligible contract $w=(w_{11},w_{10},0,0)$ targeting a known action $a_0$, $\bar{p}$ is strictly larger than zero and hence given by the closed-form solution
 \begin{equation}
	\begin{aligned}
		\bar{p}(w_{11}, w_{10})&= \frac{\sqrt{(p(a_0) w_{11}+(1-p(a_0)) w_{10})^2-2 c(a_0) (w_{11}-w_{10})}-w_{10} }{w_{11}-w_{10}}.
	\end{aligned}
	\notag
\end{equation}
In the running example with $p(a_0)=1$ and $c(a_0)= \frac{1}{4}$, the optimal wages are $w_{11}=\frac{2}{3}$ and $w_{10}=w_{01}=w_{00}=0$; the principal increases $w_{11}$ above the optimal IPE wage, $\frac{1}{2}$, to mitigate efficiency losses. In Online Appendix \ref{numericalopt_JPE}, I demonstrate via numerical optimization that optimal compensation depends on aggregate output only, i.e., $w_{11}>w_{10}=0$, whenever the surplus generated by the team, $p_0-c_0$, is sufficiently large. On the other hand, when $p_0-c_0$ is sufficiently small, monitoring individual output has value, i.e., $w_{11}>w_{10}>0$ at the optimal wage scheme.

\section{Discussion}\label{discuss}

I now discuss the assumptions made in the baseline model and their role in driving the main result. 

\subsection{The Incomplete Contracts Assumption}\label{incomplete}
 The principal's problem can be re-phrased as follows: If she must use the same contract, i.e., mapping from successes and failures into wages, given any set of actions the agents might have available, which one does the best in the sense of yielding the highest payoff guarantee?  The solution to the problem is a positive description of how a principal might write a contract in the face of structured uncertainty about the agents' environment.
 
 Two implicit assumptions underlie this formulation. First, contracts are \textit{incomplete}; they can only depend on observable successes and failures and not on the technology of the agents. Second, in line with \cite{Carroll_2015}, the principal does not possess a Bayesian prior over the agents' unknown technology.

If contracts were to be complete, then it is well known that the manager could implement the Bayesian optimal contract technology-by-technology. For instance, she could ask agents to report the true technology and, if reports disagree, punish them with a contract that always pays zero. The resulting mechanism is incentive compatible and its strict superiority over the optimal (incomplete) contracts studied does not depend on whether the principal has Bayesian or max-min uncertainty about the agents' technology. The interpretation taken in this paper, however, and in the rest of the literature on robust contracting, is that such a mechanism violates the spirit of the robustness exercise. The principal would like to avoid changing the contract she offers as the agents' environment varies, as suggested in the introductory quotation.\begin{footnote}{It is also worth noting that the main results continue to hold under the more pessimistic assumption that agents play the principal's least-preferred equilibrium within the set of Pareto Efficient Nash equilibria (see Online Appendix \ref{selection}). Under this selection assumption, more complicated, multi-stage mechanisms must be used to implement the Bayesian optimal contract.}\end{footnote}

On the other hand, there are no general, existing results on the form of optimal \textit{incomplete} contracts under \textit{Bayesian} uncertainty.\begin{footnote}{\cite{NalebuffStiglitz_1983} do consider such a model, finding conditions under which relative performance evaluation is optimal, but their production setting is not analogous to the one considered in this paper.}\end{footnote}  
 In Online Appendix \ref{app_incompletebayesian}, I show that predictions in such a setting are indeterminate. In particular, I consider a moral hazard problem in which a principal incentivizes the agents to take a costly, surplus-generating action, $a_0$, instead of a zero-cost shirking action, $a_\emptyset$, that results in failure with probability one. I posit, however, that there is another costless action, $a^*$, with intermediate productivity that is available to the agents with probability $1-\mu \in (0,1)$. If $\mu$ is sufficiently small and $a^*$ is sufficiently productive, then it is impossible to improve upon the IPE contract that always pays the agents zero. However, if $\mu$ is sufficiently large, then the optimal IPE implements $a_0$ when $a^*$ is unavailable and $a^*$ when it is available. In such cases, calibrating a JPE to the optimal IPE strictly increases the principal's payoff: these contracts enjoy the same incentive properties as the IPE contracts to which they are calibrated, i.e., they implement the same actions, while reducing expected wage payments in the scenarios in which $a^*$ is taken. This result provides an alternative, Bayesian foundation for nonlinear JPE that has not previously appeared in the literature.\begin{footnote}{It should be made clear, however, that the model studied in the Appendix is not formally analogous to the max-min model studied in the main text; the max-min problem does \textit{not} possess a saddle point and the minimax theorem does not hold (there is a duality gap). Hence, maximizing over nature’s worst-case responses yields strictly lower profits than the worst expected
payoff of the principal over all feasible production environments. The advantage of the max-min model is the sharpness of its prediction.}\end{footnote}

\subsection{The No Discrimination Assumption}\label{nodiscrim}

This paper takes as an axiom that the principal is constrained to use symmetric contracts. This is not without justification: Any asymmetric contract is \textit{discriminatory} in the sense of treating equals unequally. Hence, such contracts may be ruled out by legal considerations or --- if the principal randomizes --- ex post fairness considerations.

On the other hand, discrimination has been shown to be of value in settings in which the principal has no uncertainty about the agents' available actions and demands the contract she uses induces her preferred strategy profile as a unique Nash equilibrium. For instance, \cite{Winter_2004} shows that asymmetric contracts can be optimal even when agents are symmetric (see also \cite{Segal_2003} and \cite{Halacetal_2021}).\footnote{As \cite{Winter_2004} points out, however, if agents are assumed to play a Pareto Efficient Nash equilibrium, then there is always an optimal contract that is a symmetric IPE. This is \textit{not} the case in the setting considered here, as discussed in Section \ref{otherextensions} and shown in Online Appendix \ref{selection}.} 

Motivated by this possibility, in Online Appendix \ref{discrimination}, I identify necessary and sufficient conditions under which the optimal nondiscriminatory contract identified in the baseline model outperforms any IPE, potentially discriminatory. I show numerically that, in the running example of the analysis, it is impossible to improve upon the optimal  nondiscriminatory JPE. However, this need not always be the case; discrimination can help the principal if agents are known to share a common action set because the worst-case action for one agent constrains the worst-case action for the other.\begin{footnote}{Of course, this advantage is eliminated once the common action set assumption is relaxed.}\end{footnote} I illustrate this point via another numerical example.
 
\subsection{The Independence and Identicality Assumptions}\label{common}

As discussed in Section \ref{relatedlit}, technological independence between agents is crucial in driving the main results. However, the role of the assumption of a common action set is less clear. In Online Appendix \ref{unknown}, I consider an elaboration of Example \ref{simple_ex} in which there is a single known action and a single unknown action for each agent, potentially heterogeneous across agents. I show via a simple extension of the arguments in the main text that any optimal contract is a nonaffine JPE. 
Extending the result to the case of any number of unknown actions requires extending the characterization of worst-case payoffs of arbitrary JPE contracts in Lemma \ref{JPE_worst}. Specifically, instead of a differential equation, free-riding dynamics would instead be characterized by a nonlinear dynamical system. I conjecture that Theorem \ref{JPE_opt} would continue to hold upon extending Lemma \ref{JPE_worst} and employing the calibration and perturbation argument in Lemma \ref{JPEvIPE}, but I leave this technically challenging extension to future work.

\subsection{The Equilibrium Selection Assumption}\label{preferred}

In the model analyzed, the principal has the power to select her most preferred Nash equilibrium. In Online Appendix \ref{selection}, I consider the solution to the principal's problem under \textit{worst-case} equilibrium selection and the additional requirement that agents play a Pareto-Efficient Nash equilibrium. As non-IPE contracts tie the incentives of agents together, agents might benefit from discussing their strategies with one another, even if they cannot make binding commitments. Such communication would deem equilibria that are strictly Pareto dominated implausible, i.e., equilibria $\sigma \in \mathcal{E}(w,A)$ for which there exists another equilibrium $\sigma' \in \mathcal{E}(w,A)$ that makes each agent strictly better off. 

In this setting, all proofs in the main text hold other than that of Lemma \ref{sub_linear}, which establishes that no affine contract can outperform the optimal IPE, and that of Lemma \ref{zero}, which establishes that rewarding failure is suboptimal. Though I conjecture that the statements of these Lemmas hold, the key challenge in extending the proofs is that constant shifts in agents' payoffs can potentially affect the set of Pareto Efficient Nash equilibria. Nevertheless, it can still be established that any worst-case optimal contract is nonlinear. In addition, the proof of Lemma \ref{JPEvIPE}, which states that there exists a JPE with $w_{00}=w_{10}=0$ yielding the principal a strictly higher worst-case payoff than $V^*_{IPE}$, holds as written. This is a simple consequence of the observation that the principal's most-preferred Nash equilibrium coincides with the unique Pareto Efficient Nash equilibrium of the supermodular game induced by the contract.

\section{A General Result}\label{otherextensions}

The baseline model is kept deliberately simple in order to isolate the key intuition behind the advantage of nonaffine JPE over IPE. In addition, the two-agent, binary output setting makes it possible to completely characterize worst-case optimal contracts.\begin{footnote}{Moving beyond this case introduces tractability issues. Specifically, it is no longer possible to classify all contracts in terms of the strategic complementarity properties of the games they induce between the agents (e.g., the taxonomy of \cite{CheYoo_2001} is no longer exhaustive). Such tractability issues are present not only in the worst-case analysis of this paper, but in existing Bayesian analyses of optimal performance evaluations (see, e.g., \cite{Fleckinger_2012}).}\end{footnote} Nevertheless, the two key findings of the analysis --- that worst-case optimal contracts are nonaffine and that JPE outperforms IPE --- also hold when there are any number of agents $i=1,2,...,n$ and individual output belongs to any compact set $Y \subset \mathbb{R}_+$ with $\min(Y)=0$ and $\max(Y)=\bar{y}>0$. I outline the key ideas behind this extension.

In the many-output environment, an action, $a$, is described by an effort cost, $c(a)$, and a probability distribution over $Y$, $F(a)$. Consequently, the non-triviality assumption is that the known action set $A_0$ contains an action, $a_0$, generating strictly positive surplus: $ E_{F(a_0)}[y]-c(a_0)>0$ . For simplicity, I assume, again, that known actions are costly, i.e., if $a_0 \in A_0$, then $c(a_0)>0$.

A \textbf{contract} in this model is a function
\[ w: Y^N \rightarrow \mathbb{R}_+. \]
It is an \textbf{independent performance evaluation (IPE)} if $w(y_i, y_{-i})$ is constant in $y_{-i}$ and a \textbf{joint performance evaluation (JPE)} contract if $w(y_i, y_{-i})$ is not an IPE and is weakly increasing in $y_{-i}$ for every $y_i$. Finally, a contract is \textbf{affine} if it can be represented as a function
\[ w(y_i,y_{-i})= \alpha_0+ \alpha_i y_i+ \sum^n_{j \neq i} \alpha_j y_j,  \quad \text{ $\alpha_k \geq 0$ for $k=0,1,...,n$.}\] 

As shown in \cite{Carroll_2015}, offering each agent an IPE contract $w^*(y_i, y_{-i})= \alpha^* y_i$, where $\alpha^*= \sqrt{c(a_0)}/ \sqrt{E_{F(a_0)}[y]}$ for some $a_0 \in A_0$, yields the principal a worst-case payoff of
\[ V^*_{IPE}:= \underset{\text{$w$: $w$ is an IPE}}{\sup}~ V(w)= n \underset{w \in [0,1], a_0 \in A^0}{\max}~\left[ \left(E_{F(a_0)}[y]-\frac{c(a_0)}{w}\right) \left(1-w \right)\right]>0.\]I establish the following generalization of the main result.

\begin{theorem}\label{multiple_thm}
Suppose there are $i=1,2,...,n$ agents and output belongs to a compact set $Y$ with $\min(Y)=0< \bar{y}=\max(Y)$. Then, any worst-case optimal contract is nonaffine and there exist values of $w_0 \geq 0$ and $b > 0$ such that the nonaffine JPE contract
\[ w(y_i, y_{-i})= ( w_0 + \frac{b}{n-1} \sum^n_{j \neq i} y_j) y_i  \]
yields the principal strictly higher worst-case expected profits than $V^*_{IPE}$.
\end{theorem}

\begin{proof}
See Appendix \ref{multiple}.
\end{proof}

The steps of the proof are as follows. First,  building upon Lemma \ref{sub_linear}, I prove that any affine contract can be improved upon by an IPE. Building upon the key idea in Example \ref{simple_ex},  I then consider nonaffine JPE contracts of the form
\[w(y_i, y_{-i})= ( w_0 + \frac{b}{n-1} \sum^n_{j \neq i} y_j) y_i,\]
where $w_0 \geq 0$ is a ``base wage" and $b >0$ is a ``bonus factor" that determines how responsive wages are to the average performance of the other workers. Finally, I prove that there always exists a JPE in this class that yields the principal strictly higher worst-case expected payoffs than offering each agent the \cite{Carroll_2015}-optimal IPE. 

The key to generalizing the arguments in the two-output model to the case of many output levels is the observation that any action set $A \supseteq A_0$ can be equipped with the following total order: $a \succeq a'$ if either $E_{F(a)}[y_i] > E_{F(a')}[y_i]$, or $E_{F(a)}[y_i]=E_{F(a')}[y_i]$ and $c(a_i) \leq c(a_j)$. Under this order and under the specific class of nonaffine JPE contracts considered, any game played by the agents is supermodular. 

The key to generalizing the two-agent model to the case of multiple agents is the observation that the base wage, $w_0$, can be set equal to a number slightly smaller than the optimal IPE, $w^*$, and that the bonus factor can be calibrated according to the equation
\[ E_{F(a^*_0)}[y] \left(w_0+\frac{b}{n-1} \sum^n_{j \neq i} E_{F(a^*_0)}[y] \right)=E_{F(a^*_0)}[y] \left(w_0+b E_{F(a^*_0)}[y] \right)=w^*, \]
where $a^*_0$ is the action targeted by the optimal IPE. Under this contract, agent incentives to take less productive actions are the same as in the two-agent case. Hence, the productivity of agents in the maximal equilibrium of the worst-case supermodular game is $\bar{p}$, as in the statement of Lemma \ref{JPE_worst}, with the caveat that $\bar{p}$ is to be interpreted as the worst-case expected value of output produced by each agent. 

The so-constructed JPE possesses approximately the same incentive properties as the IPE to which it is calibrated. But, it strictly reduces the share of output each agent $i$ receives,
\[ \alpha(y_{-i}):=( w_0 + \frac{b}{n-1} \sum^n_{j \neq i} y_j),\]
when other agents $j \neq i$ are less productive. Put differently, the optimal piece rate contract is made ``flexible" in the sense of the introductory quotation of \cite{NalebuffStiglitz_1983} --- $\alpha$ is no longer a constant function of $y_{-i}$.

\section{Final Remarks}\label{finalremarks}

This paper identifies new foundations for team-based incentive pay in a canonical
moral hazard in teams setting. If a principal does not know all of the actions the agents can take, but must provide them with incentives, then nonaffine joint performance evaluation can approximate the incentive properties of any nontrivial independent performance evaluation contract, while flexibly reducing expected wage payments when agents are less productive than the principal anticipates. The worst-case analysis draws attention to these scenarios, uncovering an economic intuition that had previously gone unnoticed. 

While the focus of the article has been on the exposition of a new theoretical channel leading to the optimality of nonlinear joint performance evaluation, I conclude by discussing two applications in which this channel might be relevant. First, \cite{Herriesetal_2003} study the monthly sales of over 3,500 individual workers employed by a single telephone company. They find that there is correlation in sales performance among co-workers within the same work group. However, the authors argue that these correlations cannot be attributable to technological interdependence because sales calls are taken individually. Instead, they are potentially attributable to a combination of nonlinear joint performance evaluation and exchange of information among co-workers. Specifically, the majority of workers were compensated nonlinearly according to an increasing function of monthly individual and group revenue.

The arguments in this paper provide a plausible micro-foundation for the kind of nonlinear pay observed in  \cite{Herriesetal_2003}'s setting and in related employer-employee settings. In particular, it is plausible that group managers knew some sales tactics of the sales representatives they compensated --- for instance, representatives could always follow the telephone company's script. But, there are a myriad of less costly (but, potentially less productive) ways in which a sales representative might deviate from this script. Moreover, informational spillovers might cause these tactics to become common knowledge among workers, in which case a manager would be justifiably concerned about the compounding shirking effect that arises if incentive pay is entirely based on group performance. Thus, she might compensate workers using a mix of team-based incentive pay and individual performance bonuses. Individual performance bonuses curb individual shirking incentives, while team-based incentive pay allows the manager to reduce expected wage payments if her subordinates discover less costly, but less productive, sales tactics.

A second potential application of the joint performance evaluation result is in the literature on corporate finance and financial diversification. In this literature, \cite{Diamond_1984} shows that project financing for two independent and identical projects managed by a single risk-neutral agent protected by limited liability, e.g., an entrepreneur, is larger than if the two projects are managed by independent and identical risk-neutral agents.\begin{footnote}{My exposition here follows Section 4.2 of \cite{tirole2010theory}.}\end{footnote} The reason is that a single agent can use the income she receives from one project to pay investors when the other fails, i.e., the successful project is collateral, relaxing her limited liability constraint. This argument has been used to justify the existence of financial intermediaries, such as banks,  who oversee multiple projects and can garner a greater total supply of investment than the sum of investments garnered by individual entrepreneurs.\begin{footnote}{See the scientific background on the 2022 Sveriges Riksbank Prize in Economic Sciences in Memory of Alfred Nobel: \url{https://www.nobelprize.org/uploads/2022/10/advanced-economicsciencesprize2022-2.pdf}.}\end{footnote}

One potential issue with the line of reasoning described is that, as the number of projects a bank oversees grows, it is impractical for all projects to be monitored by the same employee. And if projects are monitored by separate employees, then diversification no longer relaxes limited liability constraints. This paper provides an alternative foundation for financial intermediaries that applies even in these cases: If investors care not about their expected returns given known monitoring actions, but instead about their worst-case expected return given potential mismanagement of their investments, then contracting with a bank can be preferable to contracting with multiple entrepreneurs. Intuitively, the bank can provide better incentives for its employees because compensation for the returns of one project can be made contingent upon the returns of the other; in cases in which both projects are poorly managed, the bank's expected wage payments decrease. Due to its superior profitability, a bank providing funds for two entrepreneurs thus receives a greater supply of investment than the sum of the funds the entrepreneurs would receive on their own.

\begin{appendix}

\section{Proofs}\label{proofs}

 \subsection{Proof of Lemma \ref{zero}}\label{proof_zero}

If $w_{11}\geq w_{01}$ ($w_{10} \geq w_{00}$), setting $w'_{11}=w_{11}-w_{01}$ and $w'_{01}=0$ ($w'_{10}=w_{10}-w_{00}$ and $w'_{00}=0$) shifts each agent's payoff by a constant. Similarly, if $w_{11} \leq w_{01}$ ($w_{10}\leq w_{00}$), setting $w'_{01}=w_{01}-w_{11}$ and $w'_{11}=0$ ($w'_{00}=w_{00}-w_{10}$ and $w'_{10}=0$) shifts each agent's payoff by a constant. It follows that any Nash equilibrium under $w$ is also a Nash equilibrium under $w'$. Since the principal's ex post payment decreases, these adjustments must (weakly) increase her worst-case payoff.

The argument in the previous paragraph immediately establishes that if $w_{11} \geq w_{01}$ and $w_{10} \geq w_{00}$, then there exists an improved contract $w'$ for which $w'_{00}=w'_{01}=0$. There are three other cases to consider: (i) $w_{01} \geq w_{11}$ and $w_{00} \geq w_{10}$ (in which case it suffices to set $w_{11}=w_{10}=0$); (ii) $w_{11} \geq w_{01}$ and $w_{00} \geq w_{10}$ (in which case it suffices to set $w_{01}=w_{10}=0$);  and (iii) $w_{01} \geq w_{11}$ and $w_{10} \geq w_{00}$ (in which case it suffices to set $w_{11}=w_{00}=0$).

\subsubsection*{$w_{01} \geq w_{11}=0$ and $w_{00} \geq w_{10}=0$}

If $w_{01} \geq 0$ and $w_{00} \geq 0$, then $w$ cannot be eligible. To wit, consider the action set $A:= A^0 \cup \{a_\emptyset\}$ where $p(a_\emptyset)=0=c(a_\emptyset)$. Then, $a_\emptyset$ is a strictly dominant strategy and so $(a_\emptyset, a_\emptyset)$ is the unique Nash equilibrium. In this equilibrium, the principal obtains a payoff $-2w_{00} \leq 0$.

\subsubsection*{$w_{11} \geq w_{01}=0$ and $w_{00} \geq w_{10}=0$}

 Under this contract, agent $i$'s payoffs satisfy increasing differences in $(a_i, a_j)$. Hence, any game this contract induces is supermodular. Moreover, fixing $a_j$, $(a_i, w_{00})$ satisfies decreasing differences. Theorem 6 of \cite{MilgromRoberts1990} then implies that the maximal equilibrium of any game $\Gamma(w, A)$, $A \supseteq A^0$, is decreasing in $w_{00}$. Since the principal's worst-case payoff either occurs when both agents succeed with probability one or in a region in which increasing the maximal equilibrium action increases the principal's payoff, the contract $w'$ with $w'_{00}=w'_{01}=w'_{10}=0$ and $w'_{11}=w_{11}$ must be such that $V(w') \geq V(w)$.

\subsubsection*{$w_{01} \geq w_{11}=0$ and $w_{10} \geq w_{00}=0$}
In this case, agent $i$'s payoff from an action profile $(a_i, a_j)$ is
\begin{equation}
\begin{aligned}
U_i(a_i, a_j; w) &=  p(a_i) (1-p(a_j)) w_{10}+ (1-p(a_i)) p(a_j) w_{01}-c(a_i)\\
&= p(a_i) \left[ w_{10}-p(a_j)(w_{10}+w_{01}) \right] + p(a_j) w_{01}-c(a_i),
\end{aligned}
\notag
\end{equation}
which satisfies decreasing differences. I show that the principal's payoff under such a contract cannot exceed $V^*_{IPE}$ when either $w_{01}>0$ or $w_{10}>0$ (at least one of these inequalities must hold for the contract to be eligible).

Let $a_\emptyset$ be the action satisfying $c(a_\emptyset)=p(a_\emptyset)=0$. Let $a^*_\epsilon$ be an action for which $c(a^*_\epsilon)=0$ and for which $p(a^*_\epsilon)$ is a fixed point of
\[T_\epsilon(p):= \begin{cases} \underset{a \in A^0 \cup \{a_\emptyset\}}{\max} \left[ p(a)-\frac{c(a)}{w_{10}-p(w_{10}+w_{01}) } \right]+\epsilon & \text{if $w_{10}-p(w_{10}+w_{01})>0$}\\ 0 & \text{otherwise} \end{cases}, \] where $\epsilon>0$ is small. To see that $T_\epsilon$ has a fixed point, notice that, for any $p \in [0,1]$, $T_\epsilon(p)$ is larger than zero (because $a_\emptyset \in A^0 \cup \{a_\emptyset\}$) and less than one if $\epsilon$ is small enough (because $A^0$ does not contain a zero-cost action that results in success with probability one by the assumption of costly known productive actions). Hence, $T_\epsilon$ is a continuous function mapping $[0,1]$ into $[0,1]$. By Brouwer's Fixed Point Theorem, it thus has at least one fixed point.

By construction, $(a^*_\epsilon,a^*_\epsilon)$ is a Nash equilibrium of $\Gamma(w,A_\epsilon)$, where $A_\epsilon:= A^0 \cup \{a^*_\epsilon, a_\emptyset\}$.
Now, consider a sequence of strictly positive values $\epsilon_1$, $\epsilon_2$,... that converges to zero and for which there is a convergent sequence of fixed points $p(a^*_{\epsilon_1})$, $p(a^*_{\epsilon_2})$,... of the mappings $T_{\epsilon_1}$, $T_{\epsilon_2}$,... . Since $[0,1]$ is a compact set, such a convergent sequence must exist. Moreover, if the limit $p^*$ satisfies $w_{10}-p^* (w_{10}+w_{01})>0$, then it must equal \[p^*:=\underset{a \in A^0 \cup \{a_\emptyset\}}{\max} \left[ p(a)-\frac{c(a)}{w_{10}-p^* (w_{10}+w_{01}) } \right].\]

I show that the principal's worst-case payoff in the limit can be no larger than what she obtains from the optimal IPE. If $p^*$ equals zero, then the principal attains less than zero profits and so lower profits than under the optimal IPE. Otherwise, let $\hat{a}_0$ denote a maximizer of $p(a)-\frac{c(a)}{w_{10}-p^* (w_{10}+w_{01}) }$ over $A^0 \cup \{a_\emptyset\}$, let $\hat{\alpha}:= (1-p^*) w_{10}$, and notice that the principal attains a payoff of

\begin{equation}
\begin{aligned}
 &2 \left[(p^*)^2 + p^* (1-p^*) (1-w_{01}-w_{10})\right]\\
 &= 2 \left[ p(\hat{a}_0)- \frac{c(\hat{a}_0)}{(1-p^*) (w_{10}+w_{01})} \right] \left[ 1-(1-p^*)(w_{10}+w_{01}) \right] \\
 & \leq 2 \left[ p(\hat{a}_0)- \frac{c(\hat{a}_0)}{(1-p^*) w_{10}} \right] \left[1-(1-p^*) w_{10}\right] \\
 &= 2 \left[ p(\hat{a}_0)-\frac{c(\hat{a}_0)}{\hat{\alpha}} \right] \left[ 1- \hat{\alpha}\right].
\end{aligned}\notag
\end{equation}
But, 
\begin{equation}
\begin{aligned}
 2 \left[ p(\hat{a}_0)-\frac{c(\hat{a}_0)}{\hat{\alpha}} \right] (1-\hat{\alpha}) & \leq 2 \underset{\alpha \in [0,1], a_0 \in A^0 \cup \{a_\emptyset\}}{\max}~\left[(1-\alpha)(p(a_0)-\frac{c(a_0)}{\alpha})  \right]\\ &= 2 \underset{\alpha \in [0,1], a_0 \in A^0}{\max}~\left[(1-\alpha)(p(a_0)-\frac{c(a_0)}{\alpha})  \right] \\ &= V^*_{IPE}, \end{aligned} \notag \end{equation} where the inequality follows because $p(\hat{a}_0)-\frac{c(\hat{a}_0)}{\hat{\alpha}} \geq 0$ for all $\hat{\alpha} \geq 0$ and the equality follows because setting $\alpha=1$ yields the principal a payoff of zero given any action in $A^0$, the payoff attained from choosing $a_\emptyset$ and any $\alpha \in [0,1]$. 
 
The previous argument establishes  that if there exists a $K$ such that, for all $k \geq K$, $(a^*_{\epsilon_k}, a^*_{\epsilon_k})$ is the unique Nash equilibrium of $\Gamma(w,A_{\epsilon_k})$, then the principal's worst-case payoff is no higher than $V^*_{IPE}$. But, other pure and mixed strategy equilibria may exist that benefit the principal, even as $k$ grows large. I now address this issue. First, consider the case in which the limit of $(a^*_{\epsilon_k})$ is $a_\emptyset$. If multiplicity arises, then there exists an action $a_0 \in A^0$ that results in success with strictly positive probability and is a weak best response to any action that succeeds with zero probability; if not, then, by Lemma \ref{lemma_super}, there would exist a $K$ such that for all $k \geq K$, $(a^*_{\epsilon_k}, a^*_{\epsilon_k})$ is the maximal Nash equilibrium of $\Gamma(w,A_{\epsilon_k})$ and hence the unique Nash equilibrium. If $p(a_0) \leq \frac{w_{10}}{w_{10}+w_{01}}$, then the principal's payoff in any equilibrium in which such an action is played with positive probability is less than zero. This follows from
\[ p(a_0) (1-w_{10}-w_{01}) \leq \frac{w_{10}}{w_{10}+w_{01}}-w_{10}<0.\] If, on the other hand, $p(a_0)>\frac{w_{10}}{w_{10}+w_{01}}$, then I can add to each $A_{\epsilon_k}$ the action $a'_0$ for which $c(a'_0)=0$ and $p(a'_0)= p(a_0)-\frac{c(a_0)}{w_{10}}$ if  $p(a_0)-\frac{c(a_0)}{w_{10}}>\frac{w_{10}}{w_{10}+w_{01}}$ and $p(a'_0)=\frac{w_{10}}{w_{10}+w_{01}}+\epsilon_k$ otherwise. In the first case, the principal attains a payoff of
\[ \left[ p(a_0)-\frac{c(a_0)}{w_{10}}\right] (1-w_{10}-w_{01}) \leq 2 \underset{\alpha \in [0,1], a_0 \in A^0}{\max}~\left[(1-\alpha)(p(a_0)-\frac{c(a_0)}{\alpha})  \right]= V^*_{IPE}.   \]  In the second case, there exists a $K$ such that for all $k \geq K$, the principal's payoff in the equilibrium $(a'_0,a^*_{\epsilon_k})$ is less than zero because the inequality in the previous displayed equation is strict. Finally, no mixed equilibria can exist in any of the cases considered since $a_\emptyset$ is a strict best response to any action larger than $\frac{w_{10}}{w_{10}+w_{01}}$ (the marginal benefit of producing succeeding with higher probability is less than zero).

Second, consider the case in which the limit of $(a^*_{\epsilon_k})$ is $p^*>0$. Any other pure or mixed Nash equilibrium of $\Gamma(w,A_{\epsilon_k})$ must involve one agent succeeding with probability $\hat{p} \geq \frac{w_{10}}{w_{10}+w_{01}}>p^*$. If not, then $p(a^*_{\epsilon_k})$ would be a best-response to the distribution $\hat{p}$ and, if $p(a^*_{\epsilon_k})$ is played, then any distribution $\hat{p}$ could not be a best-response.\begin{footnote}{The first statement follows because $p(a^*_\epsilon)$ has zero cost, profits would still be increasing in the probability with which the agent succeeds, and there are strictly decreasing differences. The second follows because $p(a^*_{\epsilon_k})$ is a strict best-response to $p(a^*_{\epsilon_k})$ by construction.}\end{footnote} However, any equilibrium in which one agent generates a distribution $\hat{p}$ must have the other play either $a_\emptyset$ (if $\hat{p} > \frac{w_{10}}{w_{10}+w_{01}}$), $a^*_{\epsilon_k}$ (only if $\hat{p} =\frac{w_{10}}{w_{10}+w_{01}}$), or a mixture between the two (again, only if $\hat{p} =\frac{w_{10}}{w_{10}+w_{01}}$); known productive actions are costly and the marginal benefit of succeeding with higher probability is less than zero (strictly so if $\hat{p}> \frac{w_{10}}{w_{10}+w_{01}}$). 
It suffices to consider the case in which $\hat{p}>\frac{w_{10}}{w_{10}+w_{01}}$. In the other two cases, introducing an action that has the same productivity as the most productive action in the support of the player's strategy that succeeds with probability $\hat{p}$, but an (arbitrarily) smaller cost, reduces the problem to this case, or alternatively, results in the equilibrium $(a^*_{\epsilon_k},a^*_{\epsilon_k})$. So, consider any action, $a_0 \in A^0$, satisfying $p(a_0) \geq \frac{w_{10}}{w_{10}+w_{01}}$ in the support of the strategy succeeding with probability $\hat{p}>\frac{w_{10}}{w_{10}+w_{01}}$. Mirroring the argument in the previous case, I can add to each $A_{\epsilon_k}$ the action $a'_0$ for which $c(a'_0)=0$ and $p(a'_0)= p(a_0)-\frac{c(a_0)}{w_{10}}+\epsilon_k$ if  $p(a_0)-\frac{c(a_0)}{w_{10}}>\frac{w_{10}}{w_{10}+w_{01}}$ and $p(a'_0)=\frac{w_{10}}{w_{10}+w_{01}}+\epsilon_k$ otherwise. These adjustments ensure that $a'_0$ is the unique best response to $a_\emptyset$ for every $k$ and so, mirroring the steps in the proof of the previous case, the principal attains a payoff no larger than $V^*_{IPE}$.

 \subsection{Proof of Lemma \ref{RPEvIPE}}\label{IPE>RPE}

Let $a_\emptyset$ be the action satisfying $c(a_\emptyset)=p(a_\emptyset)=0$. Let $a^*_\epsilon$ be an action for which $c(a^*_\epsilon)=0$ and for which $p(a^*_\epsilon)$ is a fixed point of
\[T_\epsilon(p):= \underset{a_0 \in A^0 \cup \{a_\emptyset\}}{\max} \left[ p(a_0)-\frac{c(a_0)}{p w_{11}+(1-p) w_{10}} \right] +\epsilon, \] where $\epsilon>0$ is small.\footnote{Interpret $-\frac{c(a_0)}{p w_{11}+(1-p) w_{10}}$ as zero if the denominator is zero and $c(a_0)=0$ and $-\infty$ if the denominator is zero and $c(a_0)>0$.} To see that $T_\epsilon$ has a fixed point, notice that, for any $p \in [0,1]$, $T_\epsilon(p)$ is larger than zero (because $a_\emptyset \in A^0 \cup \{a_\emptyset\}$) and less than one if $\epsilon$ is small enough (because $A^0$ does not contain a zero-cost action that results in success with probability one). Hence, $T_\epsilon$ is a continuous function mapping $[0,1]$ into $[0,1]$. By Brouwer's Fixed Point Theorem, it thus has at least one fixed point.

Now, define an action space $A_\epsilon:= A^0 \cup \{a^*_\epsilon, a_\emptyset\}$. If $A^0$ contains an action producing $y_i=1$ with probability one, consider the least costly among all of them, $\bar{a}_0$, and add to $A_\epsilon$ the action $\bar{a}_\epsilon$, where $c(\bar{a}_\epsilon)=c(\bar{a}_0)-\gamma(\epsilon)$ and $p(\bar{a}_n)=1-\frac{\gamma(\epsilon)}{2}$ for $\gamma(\epsilon) :=\frac{\epsilon (p(a^*_\epsilon) w_{11}+(1-p(a^*_\epsilon)) w_{10}}{2}$. Then, $\bar{a}_\epsilon$ strictly dominates $\bar{a}_0$ (and so any other action producing $y_i=1$ with probability one is as well) and $a^*_\epsilon$ is a strictly better reply to $a^*_\epsilon$ than $\bar{a}_\epsilon$.

I show that $(a^*_\epsilon,a^*_\epsilon)$ is the unique Nash equilibrium of $\Gamma(w, A_\epsilon)$. Notice, by construction, $(a^*_\epsilon, a^*_\epsilon)$ is a strict Nash equilibrium. Now, remove all actions producing $y_i=1$ with probability one since they are strictly dominated by $\bar{a}_\epsilon$. Upon removing these actions, $a^*_\epsilon$ strictly dominates any action smaller than it in the order $\succeq$. So, remove any actions in $\Gamma(w,A_\epsilon)$ below $a^*_\epsilon$ and denote the resulting action space by $\hat{A}$. Now, consider the profile $(\bar{a}, a^*_\epsilon)$, where $\bar{a}$ is the largest element of $\hat{A}$. Since $a^*_\epsilon$ is the unique best response to $a^*_\epsilon$ (because $(a^*_\epsilon, a^*_\epsilon)$ is a strict Nash equilibrium), the maximal best-response to $a^*_\epsilon$ is $a^*_\epsilon$. This also implies that $a^*_\epsilon$ is the minimal best-response to $\bar{a}$; if not, there exists some $\hat{a}_0 \in \hat{A}$ such that $\hat{a}_0 \succ a^*_\epsilon$ and
\[ U_i(\hat{a}_0, a_0; w)- U_i(a^*_\epsilon, a_0; w) \geq U_i(\hat{a}_0, \bar{a}; w)- U_i(a^*_\epsilon, \bar{a}; w) >0~~\text{for any $a_0 \in  \hat{A}$},\]
where the first inequality follows from the property of decreasing differences and the second from $a_0$ being the smallest best-response to $\bar{a}$. Hence, $\hat{a}_0$ strictly dominates $a^*_\epsilon$, contradicting the previous observation that $a^*_\epsilon$ is a best response to $a^*_\epsilon$. As $(a^*_\epsilon, a^*_\epsilon)$ is a fixed point of $\widetilde{BR}$, $(a^*_\epsilon,a^*_\epsilon)$ is the limit found by iterating $\widetilde{BR}$ from $(\bar{a}, a^*_\epsilon)$ or $(a^*_\epsilon, \bar{a})$ in $\Gamma(w,\hat{A})$. By Lemma \ref{lemma_sub}, it follows that $(a^*_\epsilon, a^*_\epsilon)$ is the unique Nash equilibrium of  $\Gamma(w,\hat{A})$ and hence of $\Gamma(w, A_\epsilon)$.

Now, consider a sequence of strictly positive values $\epsilon_1$, $\epsilon_2$,... that converges to zero and for which there is a convergent sequence of fixed points $p(a^*_{\epsilon_1})$, $p(a^*_{\epsilon_2})$,... of the mappings $T_{\epsilon_1}$, $T_{\epsilon_2}$,... . Since $[0,1]$ is a compact set, such a convergent sequence must exist. Moreover, its limit is the distribution \[ p(a^*)= \underset{a_0 \in A^0 \cup \{a_\emptyset\}}{\max} \left[ p(a_0)-\frac{c(a_0)}{p(a^*) w_{11}+(1-p(a^*)) w_{10}} \right].\]
Let $\hat{a}_0  \in A^0 \cup \{a_\emptyset\}$ denote the maximizer on the right-hand side and define $\hat{\alpha}:=p(a^*) w_{11}+(1-p(a^*)) w_{10}$. The principal's payoff in the unique equilibrium $(a^*_{\epsilon_k}, a^*_{\epsilon_k})$ of $\Gamma(w,A_{\epsilon_k})$ as $k$ grows large becomes arbitrarily close to \[2 \left[p(a^*) \right] \left[ p(a^*) (1-w_{11})+(1-p(a^*)) (1-w_{10}) \right]= \]
\[ 2 \left[ p(\hat{a}_0)-\frac{c(\hat{a}_0)}{\hat{\alpha}} \right] (1-\hat{\alpha}) \leq 2 \underset{\alpha \in [0,1], a_0 \in A^0 \cup \{a_\emptyset\}}{\max}~\left[(1-\alpha)(p(a_0)-\frac{c(a_0)}{\alpha})  \right], \]
where the inequality follows because $p(\hat{a}_0)-\frac{c(\hat{a}_0)}{\hat{\alpha}} \geq 0$ for all $\hat{\alpha} \geq 0$ and so I need only consider values of $\alpha$ between zero and one to maximize $(1-\alpha)(p(a_0)-\frac{c(a_0)}{\alpha})$ for any $a_0 \in A^0 \cup \{a_\emptyset\}$. But, 
\begin{equation}
\begin{aligned}
&2 \underset{\alpha \in [0,1], a_0 \in A^0 \cup \{a_\emptyset\}}{\max}~\left[(1-\alpha)(p(a_0)-\frac{c(a_0)}{\alpha})  \right]\\
=& 2 \underset{\alpha \in [0,1], a_0 \in A^0}{\max}~\left[ (1-\alpha)(p(a_0)-\frac{c(a_0)}{\alpha}) \right]\\
=&V^*_{IPE}
\end{aligned}
\notag
\end{equation}
because setting $\alpha=1$ yields the principal a payoff of zero given any action in $A^0$, the same payoff attained from choosing $a_\emptyset$ and any $\alpha \in [0,1]$.

\subsection{Proof of Lemma \ref{JPE_worst}}\label{JPEworst_proof}

\subsubsection*{Comparative Statics in Principal's Payoff}

Suppose agent $i$ succeeds with probability $p_i$. The principal's payoff given $(p_i,p_j)$ is
\[ \pi(p_i,p_j):= p_i p_j (2-2 w_{11}) + \left[ p_i (1-p_j) + (1-p_i) p_j \right] (1-w_{10}).\]
The principal's payoff is therefore increasing in $p_i$ if and only if

\[ p_j \leq \frac{1}{2} \left[ \frac{1-w_{10}}{w_{11}-w_{10}} \right].  \]
Monotonicity of $\pi(p_i,p_j)$ on $[0,1]$ thus depends on $w$: (i) if $w_{10} \geq 1$, then $\pi$ is decreasing on $[0,1]$ in $p_i$ and $p_j$; (ii) if $w_{10}<1$ and $w_{11} \leq \frac{1+w_{10}}{2}$, then $\pi(p)$ is increasing on $[0,1]$  in $p_i$ and $p_j$; and,
(iii) if $w_{10}<1$ and $w_{11} > \frac{1+w_{10}}{2}$, then $\pi(p)$ is increasing in $p_i$ if $p_j \in [0, \frac{1}{2} \left[ \frac{1-w_{10}}{w_{11}-w_{10}} \right]]$ and decreasing in $p_i$ if $p_j \in [ \frac{1}{2} \left[ \frac{1-w_{10}}{w_{11}-w_{10}} \right],1]$.

In case (i), $\pi$ is minimized when $p_i=p_j=1$, yielding the principal a payoff of \[2-2 w_{11}.\] This payoff can be achieved exactly: Consider the action set $A:=A^0 \cup \{\hat{a}\} \supseteq A^0$, where $p(\hat{a})=1$ and $c(\hat{a})=0$. Then, because $w_{11}>w_{10} \geq 1$, $\hat{a}$ is a strictly dominant strategy and so the unique Nash equilibrium of $\Gamma(w,A)$ is $(\hat{a},\hat{a})$. In case (ii), $\pi$ is minimized when the probability with which the maximal equilibrium action of $\Gamma(w,A)$, for any $A \supseteq A^0$, is as small as possible (by Observation \ref{obs_1} and Lemma \ref{lemma_super} there always exists such an action). Letting $\bar{p}$ denote the greatest lower bound on such probabilities, the principal's payoff is
\[\bar{p}^2 (2-2w_{11})+ \bar{p}(1-\bar{p}) (2-2 w_{10}). \] In case (iii), the principal's payoff is the minimum of the payoff in case (i) and case (ii),
\[ V(w)= \min\{2-2w_{11}, \bar{p}^2 (2-2w_{11})+ \bar{p} (1-\bar{p}) (2-2 w_{10})\}.\]
I identify $\bar{p}$ to complete the proof of the Lemma.

\subsubsection*{Defining $\bar{p}$}

Consider an arbitrary action $a \in A$ with cost $c(a)$ and probability $p(a)$. Let $\hat{p}(\cdot | a)$ be
a solution to the initial value problem
\begin{equation}
\begin{aligned}
\hat{p}'(t |a) &= f(\hat{p}(t|a)):= -\left[\hat{p}(t|a) w_{11} +(1-\hat{p}(t|a)) w_{10}\right]^{-1}~~\text{with}\\
 \hat{p}(0|a)&=p(a)
\end{aligned}
\notag
\end{equation}
on $D=[0,\hat{t}(a)] \times [0,p(a)]$, where $[0,\hat{t}(a)] \subseteq  [0,c(a)]$ is the largest interval on which $\hat{p}(t|a)>0$ for all $t \in [0,\hat{t}(a))$. Notice, $\hat{p}'(t|a)$ exists on $(0,\hat{t}(a))$, $\hat{p}'(t|a)<0$, and $\hat{p}''(t|a)<0$. So, $\hat{p}(\cdot|a)$ is strictly decreasing and strictly concave. Now, define \[\bar{p} := \underset{a_0 \in A^0}{\max}~~\hat{p}(\hat{t}(a_0) | a_0). \] 

\subsubsection*{$\bar{p}$ is a lower bound}

I show that $\bar{p}$ is a lower bound on the probability of the maximal equilibrium action of any game $\Gamma(w,A)$, where $A \supseteq A^0$. I begin with the following claim.

\begin{claim}[Lower Bound of a $\overline{BR}$ Path]\label{BR_LB}
Fix some game $\Gamma(w,A)$, where $A \supseteq A^0$. Let $(a_1,a_2,...,a_n)$ be the path starting from the maximal element of $A$, $a_1$, to the maximal equilibrium action, $a_n$, obtained by iterating $\overline{BR}$. If $a=a_\ell$ for some $\ell=1,...,n$, then \[p(a_n) \geq \hat{p}(\hat{t}(a)|a).\]
\end{claim}

\begin{proof}
Consider the truncated path starting at $a=a_\ell$ and ending at $a_n$. Notice that $a_{k} \in \overline{BR}(a_{k-1})$ for $k=\ell+1,...,n$ only if $p(a_{k-1})>p(a_{k})$ and,
\[ p(a_k) \left[ p(a_{k-1}) w_{11}+(1-p(a_{k-1})) w_{10} \right] -c(a_k) > p(a_{k-1}) \left[ p(a_{k-1}) w_{11}+(1-p(a_{k-1})) w_{10} \right]- c(a_{k-1})  \]
\[ \iff p(a_k) > p(a_{k-1})- \frac{c(a_{k-1})-c(a_k)}{p(a_{k-1}) w_{11}+(1-p(a_{k-1})) w_{10}}. \]
Hence, $\epsilon_k:= c(a_{k-1})-c(a_k)>0$ for any $k=\ell+1,...,n$. This implies that $\sum^n_{k=\ell +1} \epsilon_k \leq c(a)$, since $c(a_n) \geq 0$.

To show that $p(a_n) \geq \hat{p}(\hat{t}(a)|a)$, it suffices to consider the case in which $f(t,\hat{p}(t)|a)$ exists for all $t \in [0,c(a)]$ (it must always be the case that $p(a_n) \geq 0$). To show this, I need only show that $p(a_n) \geq \hat{p}(\sum^n_{k=\ell+1} \epsilon_k |a)$ because $\hat{p}(\cdot | a)$ is decreasing and so $\hat{p}(c(a) | a) \leq \hat{p}(\sum^n_{k=\ell+1} \epsilon_k | a)$. 

I prove the inequality by induction. For the base case, recall that $p(a_{\ell+1})$ must satisfy the best-response condition 
\begin{equation}
\begin{aligned}
p(a_{\ell+1}) & \geq p(a_\ell)- \frac{\epsilon_{\ell+1}}{p(a_\ell) w_{11}+(1-p(a_\ell)) w_{10}}\\
&= \hat{p}(0|a)+\hat{p}'(0|a) \epsilon_{\ell+1}\\
&\geq \hat{p}(\epsilon_{\ell+1} |a),
\end{aligned}
\notag
\end{equation}
where the last inequality follows because $\hat{p}( \cdot | a)$ is concave. 

For the inductive step, suppose $\hat{p}(\sum^m_{k=\ell+1} \epsilon_k | a) \leq p(a_m)$ for $m=\ell+1,...,K$. I show that $\hat{p}(\sum^K_{k=\ell+1} \epsilon_k+\epsilon_{K+1} | a) \leq p(a_{K+1})$. Once again, $a_{K+1}$ is a best-response to $a_K$ only if,
\begin{equation}
\begin{aligned}
p(a_{K+1}) & \geq p(a_K)- \frac{\epsilon_{K+1}}{p(a_K) w_{11}+(1-p(a_K)) w_{10}}\\
&\geq \hat{p}(\sum^K_{k=\ell+1} \epsilon_k | a)+\hat{p}'(\sum^K_{k=\ell+1} \epsilon_k | a) \epsilon_{K+1}\\
&\geq \hat{p}( \sum^{K}_{k=\ell+1} \epsilon_k+\epsilon_{K+1} | a),
\end{aligned}
\notag
\end{equation}
where the second inequality follows from the induction hypothesis and the last follows because $\hat{p}(\cdot | a)$ is concave. 
\end{proof}

Consider any finite set $A \supseteq A^0$. Let $\tilde{c}$ be the maximal cost of any action in $A$ and $\tilde{p}$ be the maximal probability. For any action $a \in A$, let $\tilde{p}( \cdot | a)$ be the solution to the initial value problem,
\begin{equation}
\begin{aligned}
\tilde{p}'(t | a)&= f(\tilde{p}(t|a))= -\left[\tilde{p}(t|a) w_{11} +(1-\tilde{p}(t|a)) w_{10}\right]^{-1} \\
\tilde{p}(\tilde{c}-c(a) | a)& =p(a),
\end{aligned}
\notag
\end{equation}
on $D=[0, \tilde{t}(a)] \times [0,\tilde{p}]$, where $[0,\tilde{t}(a)] \subseteq [0, \tilde{c}]$ is the largest interval on which $\hat{p}(t|a)>0$ for all $t \in [0,\hat{t}(a))$. Notice that $\tilde{p}(\tilde{c}-c(a)+t|a)=\hat{p}(t|a)$ for any $t \in [0,\hat{t}(a)]$, $\tilde{p}'( \cdot | a)<0$ for all $t \in [0,\tilde{t}(a))$, and  $\tilde{p}''( \cdot | a)<0$ for all $t \in [0,\tilde{t}(a))$. Moreover, the following ``no crossing" property holds; its proof is immediate upon observing that the solution to the initial value problem is unique on any interval $[0,\bar{t}]$ for $\bar{t}<\tilde{c}$, since $f'(\hat{p}(t|a))$ is bounded and exists.\footnote{See, for instance, Theorem 2.2 of \cite{CoddingtonLevinson_1955}.}

\begin{claim}[No Crossing]\label{nocrossing}
If $\tilde{p}(t | a)>\tilde{p}(t | a')$ for some $t \in [0,\tilde{t}(a)] \cap [0,\tilde{t}(a')]$, then $\tilde{p}(t' | a) \geq \tilde{p}(t' | a')$ for any other $t' \in [0,\tilde{t}(a)] \cap [0,\tilde{t}(a')]$ and so $\hat{p}(\hat{t}(a)|a) \geq \hat{p}(\hat{t}(a')|a')$.
\end{claim}

Suppose, towards contradiction, that there was a game with a maximal equilibrium action distribution $p$ satisfying $p<\bar{p}$. Then, there must exist a finite path of actions in $A$,  $(a_1,...,a_n)$, for which (i) $a_1$ is the maximal element of $A$ and $p(a_n)= p$, (ii) $p(a_1) > ... > p(a_n)$, and (iii) $a_{k} \in \overline{BR}(a_{k-1})$ (so that $c(a_1) > ... > c(a_n)$)  for $k=2,...,n$. It suffices to consider the case in which $\bar{p}> 0$, so that for any $\bar{a}_0 \in \underset{a_0}{\arg \max}~~ \hat{p}(\hat{t}(a_0) | a_0)$, $\tilde{p}'(\cdot |\bar{a}_0)$ is defined on $[0,\tilde{c}]$. Otherwise, it could never be that $p<\bar{p}$.

Now, let $a_k$ be the first action in the path $(a_1,...,a_n)$ at which $c(a_k)< c(\bar{a}_0)$. Such an action must exist. If not, then $c(a_n) \geq c(\bar{a}_0)$. So, if $p=p(a_n) < \bar{p}<p(\bar{a}_0)$, then  $(a_n,a_n)$ could not be a Nash equilibrium; $\bar{a}_0$ would be a strict best-response to $a_n$.

Consider the case in which $k=1$, so that $c(a_1)<c(\bar{a}_0)$. Then, \[\tilde{p}(\bar{c}-c(a_1) | a_1)=p(a_1) \geq p(\bar{a}_0)=\tilde{p}(\bar{c}-c(\bar{a}_0) | \bar{a}_0)>\tilde{p}(\bar{c}-c(a_1) | \bar{a}_0),\]
where the first inequality follows because $a_1$ is maximal in $A$ and the second because $\tilde{p}(\cdot | \bar{a}_0)$ is strictly decreasing. But then, $\hat{p}(\hat{t}(a_1)|a_1) \geq \hat{p}(\hat{t}(\bar{a}_0)|\bar{a}_0)$ by Claim \ref{nocrossing}. Hence, by Claim \ref{BR_LB}, \[p=p(a_n) \geq \hat{p}(\hat{t}(a_1)|a_1) \geq \hat{p}(\hat{t}(\bar{a}_0)|\bar{a}_0)= \bar{p}.\]

Consider the case in which $k>1$. Then, there exist two actions $a_{k-1}$ and $a_k$ for which $c(a_{k-1}) \geq c(\bar{a}_0) > c(a_{k})$. Notice, $p(a_{k-1}) \geq p(\bar{a}_0)$; if not and $k=2$, then $a_{k-1}$ could not have been a maximal element and, if $k>2$, then $a_{k-1}$ could not have been a best response to $a_{k-2}$ because $\bar{a}_0$ would have yielded a strictly higher payoff. Notice also that it must be the case that \[p(a_k) < \tilde{p}(\bar{c}-c(a_k) | \bar{a}_0) \leq  \tilde{p}(\bar{c}-c(\bar{a}_0) | \bar{a}_0)=p(\bar{a}_0).\] If the first inequality did not hold, then $\tilde{p}(\bar{c}-c(a_{k}) | \bar{a}_0) \leq p(a_k)=\tilde{p}(\bar{c}-c(a_{k}) | a_{k})$, in which case Claim \ref{nocrossing} implies that $\hat{p}(\hat{t}(a_k)|a_k) \geq \hat{p}(\hat{t}(\bar{a}_0)|\bar{a}_0)$. Hence, by Claim \ref{BR_LB}, it must be that $p=p(a_n) \geq \hat{p}(\hat{t}(a_k)|a_k) \geq \hat{p}(\hat{t}(\bar{a}_0)|\bar{a}_0)= \bar{p}$. The second inequality follows because $\tilde{p}(\cdot | \bar{a}_0)$ is decreasing.

I show that $\bar{a}_0$ is a weakly better response to $a_{k-1}$ than $a_{k}$, contradicting the claim that $a_k \in \overline{BR}(a_{k-1})$ (since $\bar{a}_0>a_k$). This is equivalent to showing that,
\[ p(\bar{a}_0) \left[ p(a_{k-1}) w_{11}+(1-p(a_{k-1})) w_{10} \right] -c(\bar{a}_0) \geq p(a_{k}) \left[ p(a_{k-1}) w_{11}+(1-p(a_{k-1})) w_{10} \right]- c(a_{k}) \]
\[\iff - \left[ \frac{p(\bar{a}_0) -p(a_k)}{c(\bar{a}_0)-c(a_k)} \right] \leq -\left[ \frac{1}{p(a_{k-1}) w_{11}+(1-p(a_{k-1})) w_{10}} \right].\]
Notice that,
\[ - \left[ \frac{p(\bar{a}_0) -p(a_k)}{c(\bar{a}_0)-c(a_k)} \right] \leq \frac{\tilde{p}(\bar{c}-c(\bar{a}_0)| \bar{a}_0)-\tilde{p}(\bar{c}-c(a_{k})|\bar{a}_0)}{(\bar{c}-c(\bar{a}_0))-(\bar{c}-c(a_k))} \leq \tilde{p}'(\bar{c}-c(a_{k}) | \bar{a}_0), \]
where the first inequality follows because $p(a_k) < \tilde{p}(\bar{c}-c(a_k) | \bar{a}_0)$ and the second inequality follows because $\tilde{p}( \cdot | \bar{a}_0)$ is concave. Further, 
\begin{small}
\[ -\left[ \frac{1}{p(a_{k-1}) w_{11}+(1-p(a_{k-1})) w_{10}} \right] \geq  -\left[ \frac{1}{ p(\bar{a}_0) w_{11}+(1-p(\bar{a}_0)) w_{10}} \right] = \tilde{p}'(\bar{c}-c(\bar{a}_0) | \bar{a}_0), \]
\end{small}
where the first inequality follows from $p(a_{k-1}) \geq p(\bar{a}_0)$. But, since $c(\bar{a}_0) \geq c(a_{k})$,
\[ \tilde{p}'(\bar{c}-c(a_{k}) | \bar{a}_0) \leq \tilde{p}'(\bar{c}-c(\bar{a}_0) | \bar{a}_0),\]
again by concavity of $\tilde{p}( \cdot |  \bar{a}_0)$.

\subsubsection*{$\bar{p}$ is the greatest lower bound}

I need only exhibit a sequence of action spaces $(A_n)$ for which $A_n \supseteq A^0$, $\bar{a}_n$ is the maximal Nash equilibrium action of $\Gamma(w,A_n)$, and,
\[ p(\bar{a}_n) \rightarrow \bar{p}~~\text{ as $n \rightarrow \infty$.} \] Let $\tilde{c}$ be the maximal cost of any action in $A_0$ and $\tilde{p}$ be the maximal probability. Then, define $\tilde{p}(\cdot | a)$ as before. Finally, let $\bar{a}_0 \in \underset{a_0}{\arg \max}~~ \hat{p}(\hat{t}(a_0) | a_0)$ be chosen so that $\tilde{t}(\bar{a}_0) \geq \tilde{t}(a_0)$ for all $a_0 \in A^0$.\footnote{Intuitively, $\tilde{p}(\hat{t}(a_0) | a_0)$ may equal zero for many $a_0 \in A^0$. The selection of $\bar{a}_0$ ensures that $\tilde{p}(\cdot | \bar{a}_0)$ hits zero at the largest time and therefore, invoking Claim \ref{nocrossing}, is always above the differential equations associated with other known actions.}

Suppose first that $f(t,\tilde{p}(t|\bar{a}_0))$ exists for all $t \in [0,\tilde{c}]$ so that $\tilde{p}'(\cdot|a)$ and $\tilde{p}''(\cdot|a)$ are bounded:
\[|\tilde{p}'(t|a)| \leq |\frac{p'(t|a) (w_{11}-w_{10})}{(\hat{p}(\hat{t}|a) w_{11} +(1-\hat{p}(\hat{t}|a)) w_{10})^2}|:=\kappa_1>0,\]
and,
\[|\hat{p}''(t|a)| \leq |\kappa_1 \frac{(w_{11}-w_{10})}{(\hat{p}(\hat{t}|a) w_{11} +(1-\hat{p}(\hat{t}|a)) w_{10})^2}|:=\kappa_2>0.\] Now, consider a sequence of action spaces $(A_n)$, with $A_n:= \{a^n_1, a^n_2, ..., a^n_n\} \cup  A^0$. Set $a^n_1=\tilde{p}(\underline{t}|\bar{a}_0)$, where $\underline{t} \in [0,\tilde{c}]$ is such that $\tilde{p}(\underline{t}|\bar{a}_0)=1$, and $\bar{a}_n := a^n_n$ for each $n$. Set $c(a^n_{k-1})-c(a^n_{k})=\frac{\tilde{c}}{n}:=\epsilon(n)$ for $k=2,..,n$, $\rho(n):=\frac{1}{n^2} \frac{\tilde{c}}{w_{11}+1}$, and 
 \begin{equation}
 p(a^n_k) = p(a^n_{k-1})- \frac{\epsilon(n)}{p(a^n_{k-1}) w_{11}+(1-p(a^n_{k-1})) w_{10}}+\rho(n) \tag{E} \label{Euler} \end{equation} for $k=2,...,n$. Notice,
\[- \frac{1}{n} \frac{c(a)}{p(a^n_{k-1}) w_{11}+(1-p(a^n_{k-1})) w_{10}}+\frac{1}{n^2} \frac{c(a)}{w_{11}+1}<0,\]
for $k=2,...,n$ so that $a^n_{1}>a^n_{2}>...>a^n_{n}$. Equation \ref{Euler} approximates $\tilde{p}(t | \bar{a}_0)$ on $[\underline{t},\bar{c}] \times [0,\bar{p}]$ using Euler's method with rounding error term $\rho(n)$. By the rounding error analysis of \cite{Atkinson_1989} (see Theorem 6.3 and Equation 6.2.3), since $\tilde{p}'( \cdot |a)$ is bounded by $\kappa_1>0$, and $\tilde{p}''(\cdot |a)$ is bounded by $\kappa_2>0$, it must be the case that
 \[ |p(\bar{a}_n)-\tilde{p}(\bar{c}| \bar{a}_0)| \leq \left[ \frac{e^{c(a) \kappa_1}-1}{\kappa_1} \right]  \left[\frac{\epsilon(n)}{2} \kappa_2 +\frac{\rho(n)}{\epsilon(n)} \right]. \]
 Since $\epsilon(n) \rightarrow 0$ as $n \rightarrow \infty$ and $\frac{\rho(n)}{\epsilon(n)}=\frac{1}{n} \frac{1}{w_{11}+1} \rightarrow 0$ as $n \rightarrow \infty$, the right-hand side approaches zero. Hence, $p(\bar{a}_n)$ becomes arbitrarily close to $\tilde{p}(\tilde{c}| \bar{a}_0)=\bar{p}$ as $n \rightarrow \infty$.

I need only argue that $(a^n_n, a^n_n)$ is the maximal Nash equilibrium of $\Gamma(w, A_n)$. For any $a_0 \in A^0$, $\hat{p}(\hat{t}(\bar{a}_0)|\bar{a}_0) \geq \hat{p}(\hat{t}(a_0)|a_0)$. Claim \ref{nocrossing} thus ensures that $\tilde{p}(t |\bar{a}_0) \geq \tilde{p}(t |a_0)$ for any $t \in [\underline{t},\tilde{c}]$ for which both $\tilde{p}(t |\bar{a}_0)$ and $\tilde{p}(t |a_0)$ are defined. Hence, $a^n_1=\bar{a}_0$ is the maximal element of $A_n$; if there is another action in $A^0$ that succeeds with probability one, it must have a higher cost. Finally, as Euler's method approximates $\tilde{p}(\cdot | \bar{a}_0)$ from above and there does not exist an element $a_0 \in A^0$ for which $\tilde{p}(t|a_0)>\tilde{p}(t | \bar{a}_0)$ for any $t \in [\underline{t},\bar{c}]$, $a^n_k \in \overline{BR}(a^n_{k-1})$ for each $n$ and $k=2,...n$. This implies that $a^n_n$ is the maximal Nash equilibrium action of $\Gamma(w,A_n)$.

 In the case in which $f(t,\tilde{p}(t)| \bar{a}_0)$ does \textit{not} exist for all $t \in [0, \bar{c}]$, there exists some $\bar{t} \in [0,\bar{c}]$ at which $\hat{p}(\bar{t}| \bar{a}_0)=0$, where $\tilde{p}(\bar{t}| \bar{a}_0)$ is the solution to the differential equation on $[0,\bar{t}] \times [0,p(a)]$. For any interval $[0,\hat{t}]$ such that $\hat{t}< \bar{t}$, I can mirror the argument in the case in which $f(t,\tilde{p}(t)| \bar{a}_0)$ is well-defined for all $t \in [0, \bar{c}]$ by setting $c(a^n_{k-1})-c(a^n_{k})=\frac{\hat{t}}{n}:=\epsilon(n)$ for all $k=1,..,n$ and $\rho(n):=\frac{1}{n^2} \frac{\hat{t}}{w_{11}+1}$ to show that $p(a^n_n)$ approaches $\tilde{p}(\hat{t} | \bar{a}_0)$ as $n$ goes to infinity. But $\hat{t}$ can be chosen arbitrarily close to $\bar{t}$, in which case $\tilde{p}(\hat{t} | \bar{a}_0)$ becomes arbitrarily close to $\tilde{p}(\bar{t} | \bar{a}_0)=0$. Hence, for any $\epsilon>0$, there exists a sequence of games with a maximal equilibrium action distribution $p(a^n_n)$ converging to a point in $[0,\epsilon)$ as $n$ approaches infinity. This establishes that $\bar{p}=0$ is the greatest lower bound.
 
 \subsection{Proof of Lemma \ref{JPEvIPE}} \label{JPEIPE_proof}
  
Let \[(w^*, a^*_0)  \in \underset{w \in [0,1], a_0 \in A^0}{\arg \max} (1-w)(p(a_0)-\frac{c(a_0)}{w}),\] $p^*:= p(a^*_0)$, and $c^*:= c(a^*_0)$. By the assumption of non-triviality, $p^*>c^*$ since choosing any action in $A^0$ that does not satisfy this property results in at most zero profit. By the assumption that known actions are costly, $c^*>0$ and so $w^*=\sqrt{ \frac{c^*}{p^*}} \in (0,1)$. Moreover,
\[ V^*_{IPE}=2 (1-w^*) (p^*- \frac{c^*}{w^*})<2(1-w^*).\]
Now, consider the JPE setting $w_{10}=w^*-\epsilon$, for $\epsilon>0$ small, and 
\[p^* w_{11}+ (1-p^*) w_{10}=w^* \Rightarrow w_{11}-w_{10}= \frac{\epsilon}{p^*}.\] I show that the principal obtains a strictly higher profit than $V^*_{IPE}$. Since $V^*_{IPE}=2 (1-w^*) (p^*- \frac{c^*}{w^*})<2 (1-w^*)$, I need only show that the principal obtains a higher payoff in the worst-case shirking equilibrium. 

Elementary methods show that the solution to the differential equation in Lemma \ref{JPE_worst} associated with $a^*_0$ evaluated at $c^*$ is:
 \begin{equation}
 \begin{aligned}
 \bar{p}(\epsilon):&= \frac{\sqrt{(p^* w_{11}+(1-p^*) w_{10})^2-2 c^* (w_{11}-w_{10})}-w_{10} }{w_{11}-w_{10}}
 &= p^* w^* \left( \frac{\sqrt{1-2 \epsilon}-1}{\epsilon} \right)+p^*.
 \end{aligned}
 \notag
 \end{equation}
 Using L'Hôpital's rule,
 \[ \underset{\epsilon \rightarrow 0^+}{\lim}~~\frac{\sqrt{1-2 \epsilon}-1}{\epsilon}= \underset{\epsilon \rightarrow 0^+}{\lim}~~ - (1-2\epsilon)^{-1/2}=-1. \]
 So,
\[\underset{\epsilon \rightarrow 0^+}{\lim}~~ \bar{p}(\epsilon)= p^*(1-w^*)= p^*-\frac{c^*}{w^*}. \]
In addition, for $\epsilon>0$,
\[ \bar{p}'(\epsilon)= p^* w^* \left( \frac{- (1-2\epsilon)^{-1/2} \epsilon -\sqrt{1-2\epsilon}+1}{\epsilon^2} \right). \]
Repeatedly using L'Hôpital's rule,
\begin{equation}
\begin{aligned}
\underset{\epsilon \rightarrow 0^+}{\lim}~~\frac{- (1-2\epsilon)^{-1/2} \epsilon -\sqrt{1-2\epsilon}+1}{\epsilon^2} &= \underset{\epsilon \rightarrow 0^+}{\lim}~~\frac{ -3 (1-2\epsilon)^{-5/2} \epsilon -(1-2\epsilon)^{-3/2}}{2} 
= - \frac{1}{2}.
\end{aligned}\notag
\end{equation}
So,
\[ \underset{\epsilon \rightarrow 0^+}{\lim}~~\bar{p}'(\epsilon)=-\frac{1}{2} p^*  w^* .\]
Notice, if both agents choose an action that results in success with probability $p(\epsilon)$, the principal's payoff from each agent in the shirking equilibrium is 
\[ \pi(\epsilon):= \bar{p}(\epsilon) \left[ 1- (\bar{p}(\epsilon) w_{11}+(1-\bar{p}(\epsilon)) w_{10})\right]= \bar{p}(\epsilon) \left(1 -w^* - \bar{p}(\epsilon) \frac{\epsilon}{p^*}+\epsilon\right)\]
and
\[ \underset{\epsilon \rightarrow 0^+}{\lim}~ \pi(\epsilon)= (p^*-\frac{c^*}{w^*})(1-w^*),\]
the least upper bound payoff the principal obtains from each agent within the class of IPE. Since $\bar{p}$ (as defined in Lemma \ref{JPE_worst}) is weakly larger than $\bar{p}(\epsilon)$ for every $\epsilon>0$ and profits are strictly increasing in the probability with each worker succeeds when $\epsilon>0$ is small\footnote{Simply observe that, for $\epsilon>0$ small, \[ \frac{\partial}{\partial p} \left[ p (1-w^*) + p (1-p) \epsilon \right]= (1-w^*)+(1-2p) \epsilon>0, \]
since $w^*<1$.}, it suffices to show that
\[ \partial_+ \pi(0) > 0,\]
where $\partial_+$ is the right derivative of $\pi(\epsilon)$ at $0$. For $\epsilon>0$, the derivative of $\pi$ is well-defined and equals
\begin{equation}
\begin{aligned}
 \pi'(\epsilon)
 = \bar{p}'(\epsilon) (1 -w^*- \bar{p}(\epsilon) \frac{\epsilon}{p^*}+\epsilon)+ \bar{p}(\epsilon) \left( \frac{p^*-\bar{p}(\epsilon)}{p^*}-\bar{p}'(\epsilon) \frac{\epsilon}{p^*} \right).
 \end{aligned}\notag
\end{equation}
Hence, 
\begin{equation}
\begin{aligned}
\partial_+ \pi (0)= \underset{\epsilon \rightarrow 0^+}{\lim}~~ \pi'(\epsilon) &= ( \underset{\epsilon \rightarrow 0^+}{\lim}~ \bar{p}'(\epsilon)) (1-w^*)+ (\lim_{\epsilon \rightarrow 0^+}\bar{p}(\epsilon)) \left( \frac{p^*-\lim_{\epsilon \rightarrow 0^+}\bar{p}(\epsilon)}{p^*} \right)\\
&= (-\frac{1}{2} p^* w^*)(1-w^*)+(p^* w^*) (1-w^*)\\
&= \frac{1}{2} p^* w^* (1-w^*)>0.
\end{aligned}
\notag
\end{equation}

\subsection{Proofs for Section \ref{laststeps}}\label{details_laststeps}

\subsubsection*{Existence}

A worst-case optimal JPE with $w_{00}=w_{01}=0$ solves
\begin{equation}
\begin{aligned}
\underset{w_{11}, w_{10}}{\max} & \min\{1-w_{11}, \bar{p}(w_{11}, w_{10}) \left[ \bar{p}(w_{11}, w_{10}) (1-w_{11})+ (1-\bar{p}(w_{11},w_{10})) (1- w_{10}) \right] \}\\
&\text{subject to}\\
&\bar{p}(w_{11},w_{10}) = \underset{a_0 \in A^0}{\max}~~ \hat{p}(\hat{t}(a_0; w_{11}, w_{10}) | a_0 ; w_{11}, w_{10}) \\
& 1 \geq w_{11} \geq w_{10} \geq 0,
\end{aligned}
\notag
\end{equation}
where $\hat{p}(\hat{t}(a_0; w_{11}, w_{10}) | a_0; w_{11}, w_{10})$ is defined in the statement of Lemma \ref{JPE_worst} (I now make explicit the terms that depend on the wage scheme).\begin{footnote}{I may bound $w_{11}$ above by $1$ without altering the solution set because any larger wage cannot be eligible (it yields the principal a profit of at most zero by the first argument of the objective function). I may relax the strict inequality between $w_{11}$ and $w_{10}$ to be a weak relationship without altering the solution set since I have already shown that for any wage scheme setting $w_{11}=w_{10}$ there exist wages $w_{11}>w_{10}$ that yield the principal strictly higher profits.}\end{footnote} As $\mathcal{D}:=\{(w_{11},w_{10}): 0\leq w_{10} \leq w_{11} \leq 1\}$ is a closed and bounded subset of $\mathbb{R}^2$, it is compact. Moreover, the objective function is continuous.\footnote{This follows from continuity of $\hat{p}(\hat{t}(a_0; w_{11}, w_{10}) | a_0; w_{11}, w_{10})$ (see Theorem 4.1 of \cite{CoddingtonLevinson_1955}), which in turn implies that $\bar{p}$ is continuous (since the maximum of continuous functions is continuous), which in turn implies that $\bar{p} \left[ \bar{p} (1-w_{11})+ (1-\bar{p}) (1- w_{10}) \right]$ is continuous. As $1-w_{11}$ is continuous and the minimum of two continuous functions is continuous, the result follows.} Hence, the Weierstrass Theorem ensures the existence of a solution. 

\subsubsection*{Uniqueness}

The proof of Lemma \ref{zero} shows that any contract that is not a JPE and does not set $w_{11}>0$, $w_{00}>0$, and $w_{10}=w_{01}=0$ is weakly improved upon by an IPE or RPE. Lemma \ref{RPEvIPE} and Lemma \ref{JPEvIPE} then establish that such contracts are strictly suboptimal. So, all that is left to show is that any contract setting $w_{11}>0$ and $w_{00}>0$ (with $w_{10}=w_{01}=0$) is strictly suboptimal. For this, it suffices to observe that the characterization of the principal's worst-case payoff given a JPE identified in Lemma \ref{JPE_worst} holds when replacing the law of motion in Equation \ref{diffeq} with
\begin{equation}
\hat{p}'(t) = f(\hat{p}(t)):= -\left[\hat{p}(t) w_{11} -(1-\hat{p}(t)) w_{00}\right]^{-1} \notag
\end{equation}
and setting
\[V(w)= 2 \min\{1-w_{11}, \bar{p}^2 (1-w_{11}) +(1-\bar{p})^2 (- w_{00}) \}. \] The proof of Lemma \ref{zero} establishes that setting $w_{00}=0$ yields a weak improvement for the principal. It also establishes that this improvement is strict if, given this adjustment, the principal's payoff  (from each agent) in the shirking equilibrium is smaller than $1-w_{11}$. So, I need only consider the case in which $1-w_{11}$ is strictly smaller than the principal's payoff in the shirking equilibrium. In this case, the resulting contract is strictly suboptimal; the principal could reduce $w_{11}$ by a small amount and strictly increase her payoff (because $\bar{p}$ is continuous in $w_{11}$). Hence, the original contract with $w_{00}>0$ is strictly suboptimal as well.

\subsection{Proof of Theorem \ref{multiple_thm}}\label{multiple}

I first (slightly) modify the proof of Lemma \ref{sub_linear} to establish that no affine contract can yield a higher worst-case payoff than $V^*_{IPE}$.

\begin{lemma}
Suppose there are $i=1,2,...,n$ agents and output belongs to a compact set $Y$ with $\min(Y)=0< \bar{y}=\max(Y)$. For any affine contract $w$, $V(w) \leq V^*_{IPE}$.
\end{lemma}

\begin{proof}
Suppose $w$ is an affine contract with parameters $\alpha_0 \geq 0$ and $\alpha_k \geq 0$ for all $k=1,...,n$. Consider an IPE contract $w'$ with parameters $\alpha'_0=\alpha'_j=0$  for all $j \neq i$. I claim that this contract weakly increases the principal’s worst-case payoff. First, observe that, for any $A \supseteq A_0$, the incentives of the agents are unchanged; a constant shift in an agent’s payoff holding fixed the action of the other does not affect her optimal choice of action. Hence, $\sigma \in \mathcal{E}(w,A)$ if and only if $\sigma \in \mathcal{E}(w',A)$. Second, observe that, for any equilibrium $\sigma \in \mathcal{E}(w,A)=\mathcal{E}(w',A)$, the principal's expected payoff under $w'$ is weakly larger than under $w$; her expected wage payments decrease and each agent's productivity is unchanged. Hence, $V(w',A) \geq V(w,A)$ for any $A \supseteq A^0$. It follows that
\[ V(w)=  \underset{A \supseteq A^0}{\inf}~V(w,A) \leq \underset{A \supseteq A^0}{\inf}~V(w',A)= V(w') \leq V^*_{IPE}.\]\end{proof}

 I next establish that there is a nonaffine JPE contract that yields a strictly higher worst-case payoff than $V^*_{IPE}$. For this purpose, equip any action set $A$ with the total order $\succeq$: $a \succeq a'$ if either $E_{F(a)}[y_i] > E_{F(a')}[y_i]$, or $E_{F(a)}[y_i]=E_{F(a')}[y_i]$ and $c(a_i) \leq c(a_j)$. Then, $(A,\succeq)$ is a complete lattice and any game $\Gamma(A, w)$, where $A \supseteq A_0$ and $w$ is in the class of nonaffine JPE contracts stated in the Theorem, is supermodular. In addition, the following generalization of Lemma \ref{lemma_super} applies.
 
 \begin{lemma}[\cite{Vives_1990}, \cite{MilgromRoberts1990}]\label{general_super}
Suppose $\bar{a}$ ($\underline{a}$) is the limit found by iterating $\overline{BR}$ ($\underline{BR}$) starting from $a_{\max}$ ($a_{\min}$). If $\Gamma(w,A)$ is supermodular, then it has a maximal Nash equilibrium in which all agents play $\bar{a}$.
\end{lemma}
 
 A slight modification of the two-agent calibration argument establishes the following result.

\begin{lemma}
Suppose there are $i=1,2,...,n$ agents and output belongs to a compact set $Y$ with $\min(Y)=0< \bar{y}=\max(Y)$. Then, there exist values of $w_0 \geq 0$ and $b > 0$ such that the nonaffine JPE contract
\[ w(y_i, y_{-i})= ( w_0 + \frac{b}{n-1} \sum^n_{j \neq i} y_j) y_i  \]
yields the principal strictly higher worst-case expected profits than $V^*_{IPE}$.
\end{lemma}

\begin{proof}
Let \[(w^*, a^*_0)  \in \underset{w \in [0,1], a_0 \in A^0}{\arg \max} (1-w)(E_{F(a_0)}[y]-\frac{c(a_0)}{w}),\] $p^*:= E_{F(a^*_0)}[y]$, and $c^*:= c(a^*_0)$. By the assumption of non-triviality, $p^*>c^*$ since choosing any action in $A^0$ that does not satisfy this property results in at most zero profit. By the assumption that known actions are costly, $c^*>0$ and so $w^*=\sqrt{ \frac{c^*}{p^*}} \in (0,1)$. Moreover,
\[ V^*_{IPE}=n (1-w^*) (p^*- \frac{c^*}{w^*})<n (1-w^*).\]

Now, consider the nonaffine JPE in the statement of the Lemma. Set $w_0=w^*-\epsilon$, for $\epsilon>0$ small. Choose $b>0$ to satisfy the calibration equation
\[ p^* \left(w_0+\frac{b}{n-1} \sum^n_{j \neq i} p^* \right)=p^* \left(w_0+b p^* \right)=w^*. \] Since $V^*_{IPE}<n (1-w^*)$, it suffices to show that the principal obtains a higher payoff than under the optimal IPE in the worst-case shirking equilibrium. But the principal's payoff in this equilibrium is simply $n$ times the per-agent payoff in the two-agent case, which can be seen by setting $b=w_{11}-w_{10}$ in the relevant parts of the proof of Lemma \ref{JPE_worst} and applying Lemma \ref{general_super}. Hence, from the proof of Lemma \ref{JPEvIPE}, for $\epsilon>0$ sufficiently small, the so-constructed nonaffine JPE yields the principal a strictly higher worst-case payoff.
\end{proof}

Finally, since any affine contract is outperformed by the optimal IPE and any IPE is strictly outperformed by a nonaffine JPE, any worst-case optimal contract must be nonaffine.

\bibliography{references}

\newpage
\setcounter{page}{1}
\renewcommand{\thepage}{OA-\arabic{page}}

\section{Supplementary Online Appendices}\label{online}

\subsection{Optimal JPE: Numerical Optimization}\label{numericalopt_JPE}

\begin{figure}[htbp]
    \centering
\includegraphics[scale=0.4]{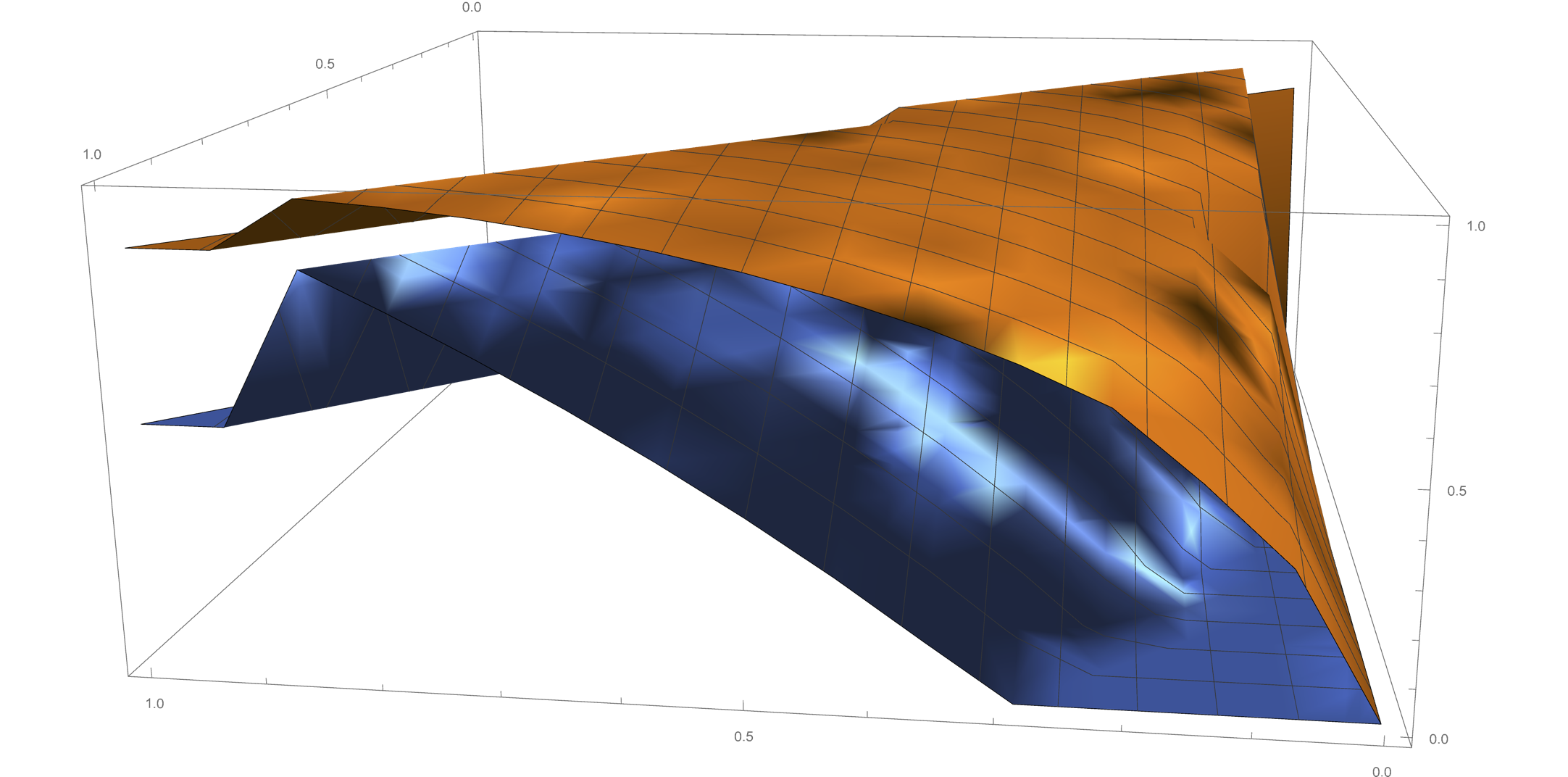}
    \caption{Optimal values of $w_{11}$ (orange surface) and $w_{10}$ (blue surface). $x$-axis: $c(a_0)$. $y$-axis: $p(a_0)$. $z$-axis: values.}
    \label{fig:optJPE}
\end{figure}

Figure \ref{fig:optJPE} depicts optimal wages $w_{11}>w_{10}\geq 0$ found by numerical optimization in Mathematica. In particular, I use the closed-form expression for the principal's payoff identified in Lemma \ref{JPE_worst} as the objective function and vary the parameters of the targeted action, $a_0$. The region in which the surplus generated by the targeted action, $p(a_0)-c(a_0)$, is large corresponds to the area surrounding the bottom-right vertex of the image box. In this region, the optimal JPE sets $w_{00}=0$ (corresponding to the blue surface) and $w_{11}>0$ (corresponding to the orange surface). Economically, monitoring individual output is of no value to the principal as she optimally bases compensation only on aggregate output. On the other hand, when $p(a_0)-c(a_0)$ is sufficiently small, both $w_{11}$ and $w_{10}$ become positive. Hence, monitoring individual output is of strictly positive value.

\subsection{Incomplete Contracts and Bayesian Uncertainty}\label{app_incompletebayesian}

In this section, I study a simple model in which the principal has Bayesian uncertainty over the set of actions available to the agents. In particular, the agents have two actions, $a_\emptyset$ and $a_0$, available with probability $\mu \in (0,1)$ and three actions, $a_\emptyset$, $a_0$, and $a^*$, available with probability $1-\mu$. $a_\emptyset$ results in success with zero probability at zero cost. $a_0$ results in success with probability $p_0>0$ at cost $c_0 \in (0,p_0)$. $a^*$ results in success with probability $p^* \in (0, p_0)$ at zero cost. The manager contemplates using one of two classes of contracts:
\begin{enumerate}
\item \textbf{Independent Performance Evaluation (IPE):} \\Pay each agent $w \geq 0$ for individual success. Pay each agent $0$ for failure.

\item \textbf{Joint Performance Evaluation (JPE):}\\ Pay each agent a wage $w_0 \geq 0$ for individual success and a team bonus $b>0$ for joint success. Pay each agent $0$ for failure.
\end{enumerate} 
I prove the following result.

 \begin{theorem}\label{incompletebayesian}
 \begin{enumerate}
     \item If $\mu$ is sufficiently small and $p^*$ is sufficiently close to $p_0$, then the optimal IPE yields the principal strictly higher expected profits than any JPE.
     \item If $\mu$ is sufficiently large, then there exists a JPE that yields the principal strictly higher expected profits than the optimal IPE.
 \end{enumerate}
 \end{theorem}
 
 \begin{proof}

 I make some preliminary observations about the optimal IPE. Observe that the optimal IPE that implements $a_\emptyset$  when the action set is $\{a_\emptyset, a_0\}$ and $a^*$ when the action set is $\{a^*, a_\emptyset, a_0\}$ is the zero contract. The principal obtains an expected payoff per agent of
 \[ (1-\mu) p^*. \] The optimal IPE implementing $a_0$ when the action set is $\{a_\emptyset, a_0\}$ and $a^*$ when the action set is $\{a^*, a_\emptyset, a_0\}$ is
     \[ w^*= \frac{c_0}{p_0}.\]
     The principal obtains an expected payoff per agent of
 \[ (\mu p_0+(1-\mu) p^*)(1- \frac{c_0}{p_0}). \] 
     Finally, the optimal IPE always implementing $a_0$ is
     \[ \hat{w}= \frac{c_0}{p_0-p^*}.\] 
    The principal obtains an expected payoff per agent of
 \[  p_0 (1- \frac{c_0}{p_0-p^*}). \] 
Any other implementation is either infeasible or suboptimal. I now separately consider the cases in which $\mu$ is small and $\mu$ is large to establish the results.

 \begin{enumerate}
     \item If $\mu$ is sufficiently small and $p^*$ is sufficiently close to $p_0$, then
     \[ (1-\mu) p^* > \max\{  \left(\mu p_0 + (1-\mu) p^* \right) (1-\frac{c_0}{p_0}), p_0 (1-\frac{c_0}{p_0-p^*}) \}. \]
     Hence, the optimal IPE puts $w^*=0$, yielding the principal a per-agent payoff of
          \[ (1-\mu) p^*. \]
      On the other hand, for a given JPE, $(w_0, b)$, the principal obtains a per-agent payoff no larger than
      \[ (1-\mu) (p^*(1-(w_0+p^* b))),\]
      when $\mu$ is sufficiently small. Because $p^* b>0$, this means the principal can do no better than the optimal IPE.
     \item 
     
    If $\mu$ is sufficiently large, then
     \[ \left(\mu p_0 + (1-\mu) p^* \right) (1-\frac{c_0}{p_0})> \max\{ p_0 (1-\frac{c_0}{p_0-p^*}) , (1-\mu) p^* \}. \] Hence, in these cases, $w^*=\frac{c_0}{p_0}$ is the optimal IPE. Now, consider a calibrated JPE with $w_0 \in (0,w^*)$ and $b=\frac{w^*-w_0}{p_0}>0$. The principal's per-agent payoff from this contract is 
     \[ \mu \left( p_0 (1-w^*) \right)+ (1-\mu) \left( p^*(1- (w_0+p^* b)) \right) > \] \[\mu \left( p_0 (1-w^*) \right)+ (1-\mu) \left( p^*(1- (w_0+p_0 b)) \right)= \underbrace{(\mu p_0+(1-\mu) p^*)(1- \frac{c_0}{p_0})}_\text{Payoff $w^*$}.   \]
     Hence, the constructed JPE strictly outperforms the optimal IPE.

 \end{enumerate}
 \end{proof}

\subsection{Discriminatory Contracts}\label{discrimination}

An \textbf{asymmetric (discriminatory) contract} is a quadruple $w^i=(w^i_{11}, w^i_{10},w^i_{01},w^i_{00}) \in \mathbb{R}^4_+$ for each agent $i=1,2$, where the first index of each wage indicates agent $i$'s success or failure and the second indicates agent $j$'s success or failure.  It is an \textbf{independent performance evaluation (IPE)} if $w^i_{y1} = w^i_{y0}$ for each agent $i=1,2$ and success or failure $y \in \{0,1\}$.

 Recall that the analysis of the optimal symmetric contract yields
\begin{equation}
	V(w^*):=\underset{w_{11}>w_{10} \geq 0}{\max}~~ \min\{1-w_{11}, \bar{p} \left[ \bar{p} (1-w_{11})+ (1-\bar{p}) (1- w_{10}) \right] \},
	\notag \end{equation}
where $\bar{p}$ is the solution to the initial value problem in the statement of Lemma \ref{JPE_worst}. The worst-case payoff from a general IPE can be identified as the value of an appropriately defined max-min problem:

\begin{equation}
    \begin{aligned}
     V^{**}_{IPE} := \quad  \underset{w_1 \geq w_2 \geq 0}{\max} \quad &\underset{a^*_1, a^*_2}{\min} \quad \left[ p(a^*_1) (1-w_1)+p(a^*_2)(1-w_2) \right] \\
     &\text{subject to}\\
     & p(a^*_1) w_1-c(a^*_1) \geq \underset{a_0 \in A^0 \cup \{a^*_1, a^*_2\}}{\max} \left[ p(a_0) w_1-c(a_0) \right]\\
     & p(a^*_2) w_2-c(a^*_2) \geq \underset{a_0 \in A^0 \cup \{a^*_1, a^*_2\}}{\max} \left[ p(a_0) w_2-c(a_0) \right],\\
    \end{aligned} \notag
\end{equation}
where $w_1$ is the wage agent 1 receives conditional upon individual success and $w_2$ is the corresponding wage for agent 2 (it is optimal to pay each agent zero for individual failure). The constraints in the minimization problem ensure that each agent $i$ has an incentive to take a worst-case unknown action $a^*_i$. No other constraints are required since one agent's optimal action is unaffected by the chosen action of the other. 

In the running example in which there is a single known action, $a_0$, with $p(a_0)=1$ and $c(a_0)=\frac{1}{4}$, the optimal wages are $w^*_{11}=\frac{2}{3}$ and $w^*_{10}=w^*_{01}=w^*_{00}=0$ yielding the principal a worst-case payoff of $V(w^*)=\frac{1}{3}$. Figure \ref{asymmetric_JPE_1} shows that, in this case, $V^{**}_{IPE}$ lies below $\frac{1}{3}$. Figure \ref{asymmetric_JPE_2} shows, however, that if $c(a_0)$ is increased to $\frac{3}{4}$, then there exist wages $w_1 \neq w_2$ that yield the principal a strictly higher worst-case payoff than under the optimal JPE, i.e. $V^{**}_{IPE}>V(w^*)$. The optimality of discrimination thus depends on the cost of effort of each agent in the principal's target action profile.

\begin{figure}
    \centering
    \includegraphics[scale=0.4]{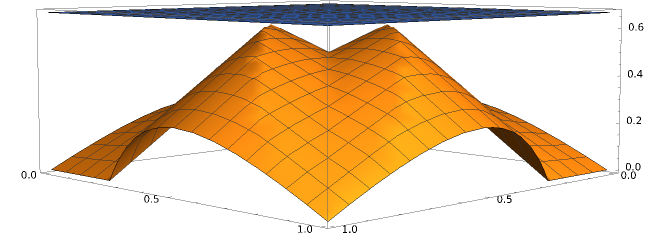}
    \caption{The orange surface represents the principal's worst-case payoff on the z-axis as discriminatory individual wages $w_1$ and $w_2$ vary. The blue surface plots the principal's worst-case payoff under the optimal nondiscriminatory JPE. Parameters: $p(a_0)=1$ and $c(a_0)=\frac{1}{4}$.}
    \label{asymmetric_JPE_1}
\end{figure}

\begin{figure}
    \centering
    \includegraphics[scale=0.4]{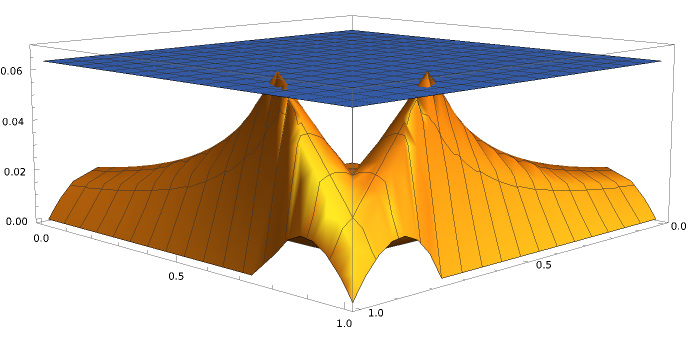}
    \caption{The orange surface represents the principal's worst-case payoff on the z-axis as discriminatory individual wages $w_1$ and $w_2$ vary. The blue surface plots the principal's worst-case payoff under the optimal nondiscriminatory JPE. Parameters: $p(a_0)=1$ and $c(a_0)=\frac{3}{4}$.}
    \label{asymmetric_JPE_2}
\end{figure}

\subsection{Asymmetric Unknown Actions}\label{unknown}

Let $\mathcal{E}(w, A_1, A_2)$ denote the set of Nash equilibria in the game induced by the contract $w$ and action sets $A_1$ and $A_2$. In addition, let
\[ V(w):= \underset{a^*_1, a^*_2 \in \mathbb{R}_+ \times [0,1]}{\min} \quad \underset{\sigma \in \mathcal{E}(w, \{a_0, a^*_1\}, \{a_0, a^*_2\})}{\max} E_{\sigma}[ y_1+y_2-w_{y_1 y_2}-w_{y_2 y_1} ]. \]
Then, the following result holds.
\begin{theorem}
Suppose there is a single known action $a_0$ with $p(a_0)>c(a_0)>0$ and each agent has at most one unknown action. Then, even if unknown actions can differ across agents, there exists a nonaffine JPE, $w$, for which $V(w)> V^*_{IPE}$.
\end{theorem}

\begin{proof}
Suppose $w$ is a nonlinear JPE with $w_{00}=w_{01}=0$. Then, $V(w)$ is the minimum of \[2-2 w_{11},\] the principal's payoff when both agents succeed with probability one, and
\[ p_1 p_2(2-2 w_{11})+(p_1 (1-p_2)+p_2(1-p_1))(1-w_{10}),\]
where
\[ p_1:= p(a_0) - \frac{c(a_0)}{(p(a_0) w_{11} + (1 - p(a_0)) w_{10})} \]
and
\[ p_2:= p(a_0) - \frac{c(a_0)}{(p_1 w_{11} + (1 - p_1) w_{10})} .\] The second expression corresponds to the principal's payoff in the limit of a sequence of games in which iterated elimination of strictly dominated strategies first removes $a_0$ for worker 1 and, second, removes $a_0$ for worker 2, leading to a unique Nash equilibrium in which worker 2 is even less productive than in the symmetric worst-case limit.

I establish the existence of a calibrated JPE, $w$, for which $V(w) > V^*_{IPE}$. Let $w^*= \sqrt{c(a_0)/p(a_0)}$ be the optimal IPE. Put $w_{10}=w^*-\epsilon$, for $\epsilon>0$, and
\[ w_{11}= \frac{w^*-(1-p(a_0)) w_{10}}{p(a_0)}. \]
It suffices to show that the right-derivative of profits with respect to $\epsilon$ evaluated at zero is strictly positive to establish the existence of a nonlinear JPE that outperforms the best IPE. For $\epsilon>0$, the derivative of profits is well-defined and equals
\[  p(a_0) w^*- \frac{c(a_0)}{(1-\epsilon)^2}. \]
Hence, the right-derivative evaluated at zero is strictly positive whenever
\[p(a_0) w^*- c(a_0)>0 \iff p(a_0)>c(a_0),\] which always holds.
\end{proof}

\subsection{Pessimistic Equilibrium Selection}\label{selection}

Denote the set of (weakly) Pareto Efficient Nash equilibria by $\mathcal{E}_{P}(w,A)$. In contrast to the model analyzed in the main text, let the principal's expected payoff given a contract $w$ and action set $A \supseteq A^0$ be given by
\[ V(w,A) := \underset{ \sigma \in \mathcal{E}_P(w,A)}{\min} E_{\sigma}[ y_1+y_2-w_{y_1 y_2}-w_{y_2 y_1} ].\]
Notice that the principal assumes the agents will play her least-preferred equilibrium. As before, the principal's worst case payoff from a contract $w$ is
\[ V(w):= \underset{A \supseteq A^0}{\inf}~V(w,A). \]  

In analyzing the nature of the solution to the principal's problem, I will need one additional definition. If $\Gamma(w,A)$ is a supermodular game and $U_i(a_i,a_j; w)$ is strictly increasing in $p(a_j)$ when $p(a_i)>0$, then $\Gamma(w,A)$ is said to exhibit \textbf{strictly positive spillovers}. The following result has been previously established in the literature.

\begin{lemma}[\cite{Vives_1990}, \cite{MilgromRoberts1990}]\label{superresult2}
Suppose $\bar{a}$ ($\underline{a}$) is the limit found by iterating $\overline{BR}$ ($\underline{BR}$) starting from $a_{\max}$ ($a_{\min}$). If $\Gamma(w,A)$ is supermodular, then it has a maximal Nash equilibrium $(\bar{a}, \bar{a})$ and a minimal Nash equilibrium $(\underline{a},\underline{a})$; any other equilibrium $(a_i,a_j)$ must satisfy $\bar{a} \succeq a_i \succeq \underline{a}$ and $\bar{a} \succeq a_j \succeq \underline{a}$. If, in addition, $\Gamma(w,A)$ exhibits strictly positive spillovers, then $(\bar{a},\bar{a})$ is the unique Pareto Efficient Nash equilibrium.
\end{lemma}

Notice that, if $w$ is a JPE for which $w_{00}=w_{01}=0$ and $A \supseteq A^0$, then $\Gamma(w,A)$ is a supermodular game exhibiting strictly positive spillovers. Hence, the principal assumes agents will play the maximal equilibrium in any game the agents play. I use this observation to establish the following result.

\begin{theorem}
Under pessimistic equilibrium selection, any worst-case optimal contract must be nonlinear and cannot be an IPE.
\end{theorem}

\begin{proof}
I first show that, for any linear contract $w$, $V(w) < V^*_{IPE}$. I first argue that any eligible linear contract must have $0<\alpha<1$. Towards contradiction, suppose $\alpha=0$. Then, the assumption of costly known actions ensures that $w$ cannot guarantee the principal more than zero in the game $\Gamma(w, A^0 \cup \{a_\emptyset\})$, where $p(a_\emptyset)=c(a_\emptyset)=0$. (In this game, each agent has a strict incentive to choose $a_\emptyset$ yielding the principal a payoff of zero.) If $\alpha \geq 1$, then the assumption of costly known actions ensures that $w$ cannot guarantee the principal more than zero in the game $\Gamma(w, A^0 \cup \{a_{\delta_1}\})$, where $p(a_{\delta_1})=1$ and $c(a_{\delta_1})=0$. (In this game, each agent has a strict incentive to choose $a_{\delta_1}$, yielding the principal a payoff of $2-2\alpha \leq 0$.)

Let $\alpha \in (0,1)$ parameterize the eligible linear contract $w$. Let $a_0 \in A^0$ be each agent's maximal equilibrium action when $A=A^0$ (since any linear contract is a JPE, such an action exists by Lemma \ref{superresult2}).  In the game $\Gamma(w', A^0 \cup \{a^*_\epsilon\})$, where $p(a^*_\epsilon)= p(a_0)- \frac{c(a_0)}{\alpha}+\epsilon$, $c(a^*_\epsilon)=0$, and $\epsilon>0$ is small, $(a^*,a^*)$ is the maximal Nash equilibrium. As $\epsilon$ approaches $0$, the principal's payoff in this equilibrium approaches
\begin{equation}
\begin{aligned}
2 \left[  \underbrace{\left(p(a_0)- \frac{c(a_0)}{\alpha} \right)}_\text{$>0$ by eligibility of $w$} \left( \left(p(a_0)- \frac{c(a_0)}{\alpha} \right) (1-\alpha)+ (1-\left(p(a_0)- \frac{c(a_0)}{\alpha} \right)) (-\alpha) \right)   \right] \\<  2 \left[  \left(p(a_0)- \frac{c(a_0)}{\alpha} \right) \left(1-\alpha \right)  \right] \\
\leq V^*_{IPE}.
\end{aligned} \notag
\end{equation}
Hence, \[V(w) \leq \quad \underset{\epsilon>0}{\inf} \quad   V(w, A^0 \cup \{a^*_\epsilon\})< V^*_{IPE}.
\]
Now, observe that the proof of Lemma \ref{JPE_worst} in Section \ref{JPEworst_proof} holds as written under worst-case Pareto Efficient Nash equilibrium selection.
Hence, there exists a JPE $w$ with $w_{00}=w_{01}=0$ for which $V(w)> V^*_{IPE}$. It follows that no worst-case optimal contract can be an IPE.
\end{proof}

 I remark that the result and proof extend to the case in which there are $n \geq 2$ agents and the set of output levels is a compact set, $Y$, with $\min(Y)=0<\max(Y)$.

\clearpage
\end{appendix}

\end{document}